%% file: main.tex
\documentclass[journal]{IEEEtran}
\usepackage{amsthm,amsmath,amssymb, amsfonts}
\usepackage{enumitem}

\theoremstyle{plain}

\theoremstyle{remark}

\theoremstyle{definition}

\setlist[description]{style=nextline,labelsep=.5em,leftmargin=0pt,font=\bfseries}
\usepackage{rotating}
\usepackage{algorithmic}
\usepackage{algorithm}
\usepackage{array}
\usepackage[caption=false,font=normalsize,labelfont=sf,textfont=sf]{subfig}
\usepackage{textcomp}
\usepackage{stfloats}
\usepackage{url}
\usepackage{verbatim}
\usepackage{graphicx}
\usepackage{cite}
\usepackage{booktabs}
\usepackage{multirow}
\usepackage{graphicx}
\usepackage{float} 
\usepackage{booktabs, multirow} 
\usepackage{soul}
\usepackage{xcolor,colortbl} 
\usepackage{changepage,threeparttable} 
\usepackage{url}
\usepackage{hyperref}
\usepackage{pgfplots}
\usepgfplotslibrary{groupplots}
\pgfplotsset{compat=1.18}

\begin{document} 
\bstctlcite{IEEEexample:BSTcontrol}

\title{Sub-Band Spectral Matching with Localized Score Aggregation for Robust Anomalous Sound Detection}

\author{Phurich Saengthong$^{1}$, Takahiro Shinozaki$^{1}$%
\thanks{$^{1}$ Department of Information and Communications Engineering, Institute of Science Tokyo, Japan (e-mail: saengthong.p.05f3@m.isct.ac.jp; website: \url{http://www.ts.ip.titech.ac.jp}).}
}


\maketitle
\begin{abstract}
Detecting subtle deviations in noisy acoustic environments is central to anomalous sound detection (ASD). A common training-free ASD pipeline temporally pools frame-level representations into a band-preserving feature vector and scores anomalies using a single nearest-neighbor match. However, this global matching can inflate normal-score variance through two effects. First, when normal sounds exhibit band-wise variability, a single global neighbor forces all bands to share the same reference, increasing band-level mismatch. Second, cosine-based matching is energy-coupled, allowing a few high-energy bands to dominate score computation under normal energy fluctuations and further increase variance. We propose BEAM, which stores temporally pooled sub-band vectors in a memory bank, retrieves neighbors per sub-band, and uniformly aggregates scores to reduce normal-score variability and improve discriminability. We further introduce a parameter-free adaptive fusion to better handle diverse temporal dynamics in sub-band responses. Experiments on multiple DCASE Task 2 benchmarks show strong performance without task-specific training, robustness to noise and domain shifts, and complementary gains when combined with encoder fine-tuning.
\end{abstract}


\newcommand{\std}[1]{{\scriptsize$\pm$ #1}}
\section{Introduction}

\input{introductionv2}

\section{Related Work}
\input{related_work}

\section{Preliminaries: Tied-Reference Global Matching}
\input{backgroundv3}

\section{Method}
\input{specmatch}

\section{Experimental Details}
\input{experimental_details}

\definecolor{lightgray}{gray}{0.92}


\section{Results}
\label{sec:results}
\input{single_experiments}

\section{Conclusion}
We present \textit{BEAM}, a band-aligned sub-band matching framework for anomalous sound detection (ASD) that replaces tied-reference global matching with band-aligned sub-band memories shared across the available reference domains, per-band nearest-neighbor retrieval, and uniform aggregation. \textit{AdaBEAM} further fuses scores from temporally mean- and max-pooled views through a parameter-free Dynamic Mean--Max rule. This design targets two sources of inflated normal-score variance: tied-reference mismatch under band-wise variability and implicit energy weighting that lets a few high-energy bands dominate matching and aggregation. Across six handcrafted and pre-trained front ends on five DCASE Task~2 datasets spanning single-domain and source--target domain-shift conditions, \textit{BEAM} increased the Overall Mean for every frontend and remained competitive with several methods requiring task-specific training. Our SDT analysis provided a sufficient condition under which variance reduction, without excessive mean-gap loss, improves detection sensitivity, while the neighborhood analysis showed that small-neighborhood local density normalization is important for consistently realizing the sub-band gains. Overall, sub-band memory modeling with localized scoring provides a practical training-free alternative to tied-reference global matching for fine-grained ASD.

The main benchmark comparisons used a fixed implementation protocol without exhaustive hyperparameter tuning. Although the average gains transferred across benchmarks, their magnitude remained dependent on the representation, dataset, and sub-band width. Future work could therefore investigate adaptive band partitioning and score calibration while preserving transfer without target-specific tuning.

\section*{Acknowledgments}
This work was supported by JTEKT Corporation.

\appendix
\label{app:dprime_reorg}
\input{appendix13}

\bibliographystyle{IEEEtran}
\bstctlcite{BSTcontrol}
\bibliography{new_refs}

\newpage

\vfill

\end{document}

%% file: introductionv2.tex
\IEEEPARstart{A}{nomalies} are defined not only by what they are, but by how they deviate from expectations. In acoustics, even subtle sound variations can reveal hidden faults that would otherwise go unnoticed. This motivates anomalous sound detection (ASD) for industrial condition monitoring, where an ASD system assigns an anomaly score to each clip and raises an alarm when the score crosses a threshold \cite{Koizumi_8081297, Uematsu2017AnomalyDT, purohit_mimii_2019, koizumi_toyadmos_2019, koizumi_description_2020}. Because abnormal events are rare and diverse, ASD is typically trained in a semi-supervised setting using only normal data \cite{chandola2009anomaly, wilkinghoff2024audio}. In deployment, shifts in background noise, environment, or operating conditions can widen the normal-score distribution and destabilize fixed thresholds. For example, a normal machine recording can exceed a threshold tuned on training data when the noise floor rises, causing false alarms and frequent retuning.

Many recent state-of-the-art (SOTA) ASD systems \cite{jiang_anopatch_2024, 10887974, SaengthongSCITOK2025, saengthong2026retaining, bing2025explore, local_den_norm_kevin, wilkinghoff2025density} adopt a \emph{global matching} design, where each clip is summarized by a single clip-level embedding and compared with a reference memory via nearest-neighbor (NN) search, most commonly using cosine similarity. In this framework, domain shifts caused by changing machine operating conditions or noise, and the resulting threshold instability, are often addressed through score normalization or calibration \cite{10887974, SaengthongSCITOK2025, local_den_norm_kevin, wilkinghoff2025density}.

GenRep (``generic representation'')-style systems \cite{10887974, SaengthongSCITOK2025, wilkinghoff2025density} show that large-scale, general-domain pre-trained audio encoders can deliver strong performance and robustness under domain shift while remaining competitive with target-machine encoder fine-tuning \cite{Nishida_arXiv2025_01, SaengthongSCITOK2025}. This makes global matching a compact, practical \emph{training-free} baseline when task-specific training is infeasible, requiring only the storage of reference embeddings for rapid deployment across diverse machine operating and recording conditions without target-machine training. Here, \emph{training-free} means that neither the target-machine encoder nor the scoring model is optimized on the reference normal clips, in contrast to fine-tuning, which updates encoder parameters using target-machine data or auxiliary objectives.

In signal detection theory (SDT) \cite{green1966signal}, detection degrades when the normal-score distribution broadens relative to the anomalous--normal separation. We analyze a failure mode of global cosine matching that increases normal-score variance through two mechanisms. Consider normal machine sounds whose underlying state is unchanged, but whose test recordings have low SNR and background-noise conditions different from those represented in the references stored in the memory bank. Let the clip embedding be decomposed into sub-bands (i.e., fixed feature partitions), and suppose that normal recording-condition variation perturbs these sub-bands unequally. Global cosine matching then selects a single nearest-neighbor reference using the full clip embedding. Because cosine similarity in this setting is magnitude-coupled, higher-magnitude sub-bands contribute more to both reference selection and final score computation. Thus, (i) the best whole-embedding neighbor need not be the best match for each sub-band, and (ii) a few high-magnitude sub-bands can disproportionately influence the final score. Both effects can widen the normal-score distribution and make anomaly decisions less reliable even when the underlying machine state remains normal.

Motivated by these two effects, we propose \textit{BEAM} (Band-wise Equalized Anomaly Measure), a simple $k$-NN-based scoring method that replaces single-reference global matching with band-aligned sub-band matching and uniform aggregation. Given a clip-level representation with an ordered feature axis, \textit{BEAM} stores reference descriptors in a band-aligned memory bank of sub-band embeddings shared across reference domains; at inference, each sub-band queries its corresponding bank and the band scores are aggregated uniformly. This per-band retrieval decouples reference selection across sub-bands, reducing tied-reference mismatch, while uniform aggregation avoids global energy-coupled weighting and prevents high-magnitude sub-bands from dominating the final score, thereby stabilizing normal scores. We further introduce \textit{AdaBEAM}, which integrates Dynamic Mean--Max (DMM) scoring by maintaining temporally mean-pooled and max-pooled sub-band memories, scoring both views, and fusing the clip-level scores with a simple parameter-free rule. Finally, we formalize how sub-band matching and uniform aggregation can reduce normal-score variance and provide an SDT-based sufficient condition showing when this variance gain outweighs any accompanying change in mean separation.

We evaluate \textit{BEAM} across several feature families, including both handcrafted and pre-trained clip representations. The temporally pooled features comprise Log-Mel, MFCC, and deep spectrogram embeddings from pre-trained audio encoders. We also include LPC spectra, which have been explored in prior ASD work \cite{GleichmannTNT2024,saengthong2026memnmf} but remain relatively underused compared with spectrogram-derived features. Their ability to capture smooth spectral envelopes and salient resonances makes them a natural match for our focus on fine-grained spectral structure. On DCASE Task~2 benchmarks \cite{koizumi_description_2020,dohi_description_2022,dohi_description_2023,Nishida_arXiv2024_01,Nishida_arXiv2025_01} under the standard evaluation protocols, \textit{BEAM} and its adaptive variant \textit{AdaBEAM} achieve strong Official Scores among training-free systems when combined with frozen, pre-trained audio encoders. Moreover, \textit{BEAM} remains effective with handcrafted features and can surpass task-specifically trained autoencoder baselines. Across the spectrogram-derived representations, \textit{AdaBEAM} generally provides additional gains over \textit{BEAM} by fusing complementary temporally mean-pooled and max-pooled views via DMM. Beyond the training-free setting, \textit{BEAM} and \textit{AdaBEAM} also generally improve over the corresponding global-matching baselines when the encoder is fine-tuned with task-specific objectives on target-machine data, although our primary focus is the training-free regime.

We make the following contributions:
\begin{itemize}
    \item \textbf{Variance-focused analysis of global matching:} we characterize how tied-reference retrieval and energy-coupled aggregation contribute to normal-score variance inflation under domain shift and noise.
    \item \textbf{BEAM scoring:} we introduce \textit{BEAM}, which performs band-aligned sub-band retrieval with uniform aggregation, and describe its variant \textit{AdaBEAM}, which fuses mean- and max-pooled views via Dynamic Mean--Max (DMM) scoring.
    \item \textbf{SDT-based condition for improved sensitivity:} we derive a sufficient condition under which variance reduction yields a net gain in detection sensitivity $d'$ relative to global cosine matching.
    \item \textbf{Empirical validation on DCASE Task~2:} the proposed methods achieve strong training-free performance across different feature types, remain competitive with several methods that rely on task-specific training, and also improve fine-tuned global-matching baselines in most evaluated settings.
\end{itemize}

%% file: related_work.tex
This work builds on ASD systems developed in the DCASE community, particularly Task~2 \cite{koizumi_description_2020, Kawaguchi2021, dohi_description_2022, dohi_description_2023, Nishida_arXiv2024_01, Nishida_arXiv2025_01, wilkinghoff2025handlingdomainshiftsanomalous}. Early task-specific approaches fit generative models to normal data and score anomalies by reconstruction error or likelihood \cite{Koizumi_8081297, koizumi_spidernet_9053620, hayashi2020conformer, suefusa_anomalous_2020, wichern_anomalous_2021, 10095568, zavrrtanki_10447941, Harada_arXiv2023_01, Harada_EUSIPCO2023_01, wavenet_8553423, Giri2020_made, lee2021robust, guan_10096356, glowaff_9414662, Dohi2021DisentanglingPP, saengthong2026memnmf}. Later discriminative and outlier-exposure (OE) objectives \cite{hendrycks2019oe} learned noise-robust embeddings, often using machine-type metadata, for lightweight back-ends such as k-NN \cite{eskin_geometric_2002}. Such embeddings are trained from scratch \cite{Giri2020, wilkinghoff_sub-cluster_2021, liu_anomalous_2022full, wilkinghoff_self-supervised_2023_full, zhangOutlierawareInlierModeling2023, Tanaka_10446288, wilkinghoff_why_2024, 10447859, wilkinghoff2025density, 10890020} or obtained by fine-tuning pre-trained encoders \cite{10290003, 10447183, han_10095398, jiang_anopatch_2024, anopatch_lora, 10889514, bing2025explore, wilkinghoff2025density}. More recently, large-scale self-supervised encoders have achieved competitive Task~2 performance without OE training or fine-tuning \cite{10887974, wu2025towards, SaengthongSCITOK2025, saengthong2026retaining}. Most of these works focus on representation quality rather than anomaly scoring.

Large-scale audio encoders pre-trained on datasets such as AudioSet provide transferable representations \cite{gemmeke_audio_2017, openl3, gong21b_interspeech, niizumi2022masked, chen_beats_2022, dinkel2023ced, niizumi2024m2dx, niizumi2024m2d-clap, ijcai2024p421, alex2025sslam}, particularly with temporal mean pooling \cite{niizumi2022masked, 10887974, saengthong2026retaining, wilkinghoff2025density}. GenRep \cite{10887974} uses mean-pooled BEATs \cite{chen_beats_2022} with MemMixup and domain-score normalization, whereas AnoPatch \cite{jiang_anopatch_2024} fine-tunes BEATs and applies attentive statistics pooling to patch representations. Both ultimately apply tied-reference $k$-NN to clip-level embeddings with domain-specific memories. In contrast, \textit{BEAM} preserves frequency resolution after temporal pooling, retrieves neighbors from domain-agnostic band-aligned memories shared across domains, and uniformly aggregates the sub-band scores. It can also incorporate LDN \cite{local_den_norm_kevin,wilkinghoff2025density,SaengthongSCITOK2025}; neighborhood-size robustness is further studied in \cite{wilkinghoff2026mind}. Rather than normalizing one clip-level global score, we normalize each sub-band distance before aggregation.

DMM is related to temporal-pooling methods that balance robustness and sensitivity. Combining pooling with GMMs performs best when mean or max pooling is selected per machine type \cite{wilkinghoff_sub-cluster_2021}, while time-weighted pooling balances mean and max effects but still requires machine-specific tuning \cite{guan_10096356}. Relative deviation pooling (RDP), introduced in concurrent work \cite{wilkinghoff2026temporal}, uses non-learned, attention-like weights controlled by an exponent selected per embedding on development data. In contrast, the default DMM uses fixed arithmetic-mean fusion of the \textit{BEAM} scores from mean- and max-pooled views, requiring neither development-set tuning nor machine-specific pooling selection.

Sub-band matching is conceptually related to frequency-dependent audio modeling. Prior work has used sub-band evidence to interpret ASD decisions \cite{mai2022explaining}, frequency-wise attention pooling to fine-tune pre-trained ASD models \cite{jiang2025finetune}, and frequency-dependent architectures for ASD and sound-event detection \cite{fan2025fisher,nam22_interspeech,guan2023subband}. In contrast, \textit{BEAM} operates at the anomaly-scoring stage rather than the representation-learning stage, applying band-aligned retrieval and uniform score aggregation to temporally pooled descriptors.

Finally, \textit{BEAM} with pre-trained audio transformers is conceptually related to patch-based industrial image anomaly detection \cite{bergman_deep_2020, Cohen2020SubImageAD, roth_towards_2022, Jeong2023}. We collapse the time axis of ViT-style token embeddings \cite{dosovitskiy2021an, Jeong2023} into frequency-indexed descriptors, so each ``patch'' summarizes a frequency range over the entire clip. We therefore aggregate uniformly across bands, whereas many patch-based image methods use maximum spatial pooling to emphasize localized defects.

%% file: backgroundv3.tex
\label{sec:baseline-global}

This section formalizes the tied-reference global-matching formulation that \textit{BEAM} builds on. Each audio clip is mapped to a clip-level representation, and anomaly evidence is computed from a single nearest-neighbor match in that representation space \cite{wilkinghoff_sub-cluster_2021, 10887974, local_den_norm_kevin, wilkinghoff2025density}. We first define the clip representations used in this work, then summarize the global nearest-neighbor score used in recent systems (e.g., GenRep-style pipelines \cite{10887974, wilkinghoff2025density}), and finally analyze the structure and limitations of exact global cosine scoring.

\subsection{Clip-Level Representations}
\label{sec:freq-des}

We represent each audio clip by a vector $f\in\mathbb{R}^F$ whose coordinates follow an ordered feature axis (e.g., frequency bins, spectral bands, or ordered coefficients). We assume this axis can be partitioned into meaningful sub-bands. In this work, we instantiate $f$ using three front-end families: (i) handcrafted frame-level spectral/cepstral features followed by temporal pooling (Log-Mel filterbank energies and MFCCs), (ii) spectrogram-based embeddings from large-scale pre-trained encoders that produce patch-level features before pooling, and (iii) LPC spectra, obtained directly from all-pole signal modeling rather than from frame-level temporal pooling.

\subsubsection{Temporally pooled spectral and cepstral features}

Let $X \in \mathbb{R}^{T \times F}$ denote a frame-level feature matrix, where $T$ is the number of frames and $F$ is the feature dimension. For spectral features such as Log-Mel filterbank energies \cite{wilkinghoff_sub-cluster_2021}, index $j$ corresponds to a mel-frequency band. For cepstral features such as MFCCs, index $j$ corresponds to an ordered cepstral coefficient. Although the MFCC axis is not a frequency-band axis, it is ordered and can be partitioned; therefore, we include MFCCs as an additional instantiation of $f$. We obtain the clip-level vector $f\in\mathbb{R}^F$ by temporal mean pooling, $f(j)=T^{-1}\sum_{t=1}^{T}X(t,j)$, or temporal max pooling, $f(j)=\max_t X(t,j)$, for $j=1,\dots,F$. Both operations summarize the recording over time while preserving the ordered feature axis (frequency for Log-Mel and coefficient index for MFCCs).

\subsubsection{Spectrogram embeddings from pretrained encoders}
For Transformer-based encoders \cite{NIPS2017_3f5ee243} such as BEATs \cite{chen_beats_2022}, the input spectrogram is divided into a grid of time--frequency patches, each mapped to a $D$-dimensional embedding. This yields patch embeddings $h\in\mathbb{R}^{L\times D}$ with $L=T_p\cdot F_p$, where the subscript $p$ denotes patch-grid positions and $T_p$ and $F_p$ are the numbers of patch positions along time and frequency, respectively. We reshape $h$ into $T_p\times F_p\times D$, group embeddings by frequency index, and apply temporal pooling over the $T_p$ patches within each frequency band to obtain $F_p$ band descriptors in $\mathbb{R}^{D}$. Concatenating these per-band embeddings yields a clip-level representation $f\in\mathbb{R}^{F_p\cdot D}$ whose coordinates remain organized by frequency band. This treatment is consistent with prior uses of pooled patch embeddings for ASD, which commonly use temporal mean pooling \cite{10887974,wilkinghoff2025density}. Encoders fine-tuned on reference-machine data using auxiliary objectives \cite{jiang_anopatch_2024,10890020,fujimura2025asdkit} can be used in the same way.

\begin{figure*}[!t]
    \centering
    \includegraphics[width=0.89\textwidth]{figs/beam.pdf}
    \caption{Overview of \textit{BEAM} and \textit{AdaBEAM}. (a) Memory construction: normal reference clips are mapped to handcrafted or deep features, decomposed at fixed sub-band positions, and stored in the corresponding band-aligned memories $\{\mathcal M_j\}_{j=1}^{N_b}$. (b) \textit{BEAM} inference: each test sub-band queries only its corresponding memory, and the resulting band scores are uniformly aggregated into a clip-level anomaly score. (c) \textit{AdaBEAM}: parallel temporally mean- and max-pooled views are scored against their respective sub-band memories, and the two clip-level scores are combined by parameter-free Dynamic Mean--Max (DMM) fusion.}
    \label{fig:method}
\end{figure*}

\subsubsection{LPC spectrum}
\label{sec:lpc-prelim}
LPC is a classic spectral-envelope model \cite{saito1967theoretical,lpc_review_1451722} that has been used in ASD with learned and reconstruction-based models \cite{GleichmannTNT2024,saengthong2026memnmf}. Fitting an all-pole model to the waveform yields a smooth envelope that can emphasize resonances and timbre-related differences \cite{saengthong2026memnmf,nishida2024timbredifferencecapturinganomalous}. Given a waveform segment $s(n)$, we estimate LPC coefficients $\{a_k\}_{k=1}^{p}$ (e.g., via Burg’s method~\cite{burg1975maximum,andersen1974calculation,saengthong2026memnmf}) and sample the all-pole magnitude response as $f[m]=\left|1/(1-\sum_{k=1}^{p}a_ke^{-j\omega_mk})\right|$, $m=1,\dots,F$, yielding an LPC spectrum $f\in\mathbb R^F$.

\subsection{Tied-Reference Global Matching}
\label{sec:global-knn}
GenRep-style training-free ASD systems \cite{10887974, SaengthongSCITOK2025, wilkinghoff2025density} typically score a test clip by global nearest-neighbor matching against a normal reference memory, optionally followed by score normalization (e.g., domain-wise normalization \cite{10887974, SaengthongSCITOK2025} or local density normalization \cite{local_den_norm_kevin, wilkinghoff2025density}). Here, we analyze the core tied-reference global-matching score before such normalization. Given a clip-level representation $f_y\in\mathbb{R}^F$ and a normal reference memory $\mathcal M_{\text{glob}}=\{f_i\}_{i=1}^R$ (with $R$ normal reference clips), the global-matching pipeline scores test clip $y$ using its single nearest neighbor under cosine similarity:
\begin{equation}
\label{eq:s-global}
\begin{aligned}
i^\ast(y) &\,=\, \min\operatorname*{arg\,max}_{i\in\{1,\ldots,R\}}
\big\langle \widehat f_y,\,\widehat f_i\big\rangle,\\
S_{\mathrm{glob}}(y)
&\,:=\, d_{\mathrm{NN}}(f_y,\mathcal M_{\text{glob}})
\,=\, \tfrac{1}{2}\Bigl(1-\big\langle \widehat f_y,\,\widehat f_{i^\ast(y)}\big\rangle\Bigr),
\end{aligned}
\end{equation}
where $d_{\cos}(u,v):=\tfrac12(1-\langle\widehat u,\widehat v\rangle)$ denotes the pairwise cosine distance, so $S_{\mathrm{glob}}(y)=d_{\cos}(f_y,f_{i^\ast(y)})$. Here, $\widehat x := x/\|x\|_2$ for $\|x\|_2>0$, and $\widehat 0:=0$. Ties in the nearest-neighbor search are resolved by
choosing the smallest reference index, which makes $i^\ast(y)$
single-valued.

\subsection{Global Band Cosine Decomposition}
\label{sec:global-decomp}
To analyze how this tied-reference score behaves across the ordered feature axis, we rewrite $S_{\mathrm{glob}}(y)$ using an arbitrary fixed partition into $N_b$ disjoint sub-bands. The partition need not be optimized for this derivation; in the experiments, the concrete windows are specified in Sec.~\ref{sub-sec:impl}. Let $w_{y,j}$ and $w_{i^\ast(y),j}$ denote the sub-band-$j$ slices of test clip $y$ and its global nearest neighbor $i^\ast(y)$, and let $s_j(y) := \langle \widehat w_{y,j},\,\widehat w_{i^\ast(y),j}\rangle$ denote the local sub-band similarity. When $\|f_y\|_2\|f_{i^\ast(y)}\|_2>0$, define the sub-band energy weights $\beta_j(y):= \|w_{y,j}\|_2\|w_{i^\ast(y),j}\|_2 /(\|f_y\|_2\|f_{i^\ast(y)}\|_2)$; otherwise, set $\beta_j(y):=0$. Let $\rho(y):=\sum_{j=1}^{N_b}\beta_j(y)$ and $\tilde\beta_j(y):=\beta_j(y)/\rho(y)$ when $\rho(y)>0$ (or $1/N_b$ when $\rho(y)=0$). Additivity of the inner product across disjoint coordinate blocks then gives
\begin{equation}
\label{eq:glob-decomp-sum}
S_{\mathrm{glob}}(y)
=\frac12\left(1-\rho(y)\sum_{j=1}^{N_b}\tilde\beta_j(y)s_j(y)\right).
\end{equation}
Thus, global cosine scoring is an energy-weighted aggregation of local sub-band similarities.

This rewriting separates two limitations of tied-reference global matching, both of which can be amplified by domain shift or added noise. For comparison, define the band-wise best index $i_j^\ast(y):=\min\operatorname*{arg\,max}_{i\in\{1,\ldots,R\}}\langle\widehat w_{y,j},\widehat w_{i,j}\rangle$. In general, $i_j^\ast(y)$ need not equal the globally selected index $i^\ast(y)$. Thus, before Eq.~\eqref{eq:glob-decomp-sum} is evaluated, all bands have already been forced to share $i^\ast(y)$, and $s_j(y)$ may reflect a global compromise rather than the best available match for band $j$. Second, the aggregation is energy-coupled through the data-dependent weights $\tilde\beta_j(y)$, which can shift with recording conditions and concentrate mass on a few bands. For squared Euclidean distance, the score decomposes additively as $\|f_y-f_i\|_2^2=\sum_j\|w_{y,j}-w_{i,j}\|_2^2$ for a fixed disjoint partition and reference $i$; hence the tied-reference limitation remains, but the cosine-specific energy-weighted similarity term is replaced by unnormalized squared band errors. This difference can make Euclidean matching sensitive to absolute feature scale, whereas cosine matching emphasizes angular similarity with magnitude-coupled band weights. Both cases motivate band-wise retrieval and explicit aggregation choices.

%% file: specmatch.tex
\label{sec:beam}
We propose \textit{BEAM} to mitigate the variance inflation of tied-reference global matching (Sec.~\ref{sec:global-knn}) through band-aligned retrieval and uniform aggregation. Fig.~\ref{fig:method}(a) shows how normal reference clips are encoded, decomposed at fixed sub-band positions, and stored in one aligned memory $\mathcal M_j$ per position. At inference, as shown in Fig.~\ref{fig:method}(b), each query sub-band is compared only with its corresponding memory, and the resulting scores are uniformly aggregated. Fig.~\ref{fig:method}(c) extends this pipeline to \textit{AdaBEAM}, which computes parallel scores from temporally mean- and max-pooled views and combines them through Dynamic Mean--Max (DMM) fusion. Secs.~\ref{sec:mem-sub}--\ref{sec:dmm} define these three stages formally.
 
\subsection{Memory Bank of Sub-Band Embeddings}
\label{sec:mem-sub}
Given a clip representation $f\in\mathbb{R}^{F}$ (obtained either by temporal pooling of a time--frequency feature or directly from an LPC spectrum; Sec.~\ref{sec:freq-des}), \textit{BEAM} uses a simple sliding-window partition of the ordered feature axis. Each reference descriptor $f_i$ is sliced into $N_b$ contiguous sub-bands:
\[
\mathbf{w}_{i,j}=f_i[\mathcal I_j]\in\mathbb{R}^{C},\qquad j=1,\dots,N_b,
\]
For regular-grid positions,
\[
\mathcal I_j=\{1+(j-1)s,\dots,1+(j-1)s+C-1\},
\]
where $C$ and $s$ are the window size and stride. The regular grid contains $1+\lfloor(F-C)/s\rfloor$ windows. If it does not cover the final coordinates, we append the end-aligned window $\{F-C+1,\dots,F\}$ as described in Sec.~\ref{sub-sec:impl}, and $N_b$ denotes the resulting total. The special cases are intuitive: $C=F$ recovers global matching with one band, while small $C$ values produce finer localized matching at the cost of more band scores. We maintain a band-aligned memory bank for each position,
\[
\mathcal M_j := \bigl\{\mathbf w_{i,j}\bigr\}_{i=1}^{R},\qquad j=1,\dots,N_b,
\]
where $R$ denotes the number of normal reference samples. We denote the full memory as $\mathcal M=\{\mathcal M_j\}_{j=1}^{N_b}$. This full-memory version is retained for comparability with recent $k$-NN ASD systems. Coreset selection could reduce storage and search costs when normal references are redundant. Under domain shift, its effect may depend on how compression changes source--target coverage, nearest-neighbor retrieval, and the local neighborhoods used by LDN; we do not evaluate this interaction systematically here. We therefore leave broader coreset design and evaluation to future work.

\subsection{Local Matching and Clip-Level Aggregation}
\label{sec:sub-match-agg}
Given a test clip $y$, we form its sub-band descriptors $\{\mathbf w_j(y)\}_{j=1}^{N_b}$ and score each band by a band-restricted nearest-neighbor distance:
\begin{equation}\label{eq:sub-nn}
d_{\text{sub}}(\mathbf w_{j}(y))=\min_{\mathbf w\in\mathcal M_j} d_{\cos}\!\left(\mathbf w_{j}(y),\mathbf w\right),\qquad j=1,\dots,N_b,
\end{equation}
where $d_{\cos}(\cdot,\cdot)$ is the pairwise cosine distance defined in Sec.~\ref{sec:global-knn}. The clip-level score is the uniform aggregation of band-wise distances:
\[
S_{\mathrm{sub}}(y):= \frac1{N_b}\sum_{j=1}^{N_b} d_{\text{sub}}(\mathbf w_{j}(y)).
\]

The score $S_{\mathrm{sub}}(y)$ already defines a valid clip-level anomaly measure and directly instantiates the variance-motivated design in Sec.~\ref{sec:global-decomp}. In practice, however, the raw sub-band distances $\{d_{\text{sub}}(\mathbf w_j(y))\}$ can exhibit band- and domain-dependent scales (e.g., different local noise floors). To improve cross-band comparability before aggregation, we apply local density normalization (LDN)~\cite{local_den_norm_kevin} at the sub-band level. For $K\geq1$, define the normalized distance as
\begin{equation}
\label{eq:sub-ldn}
d^{\mathrm{norm}}_{\mathrm{sub}}(\mathbf w_{j}(y))
:=
\min_{\mathbf w\in\mathcal M_j}
\frac{d_{\cos}\!\left(\mathbf w_{j}(y),\,\mathbf w\right)}
{\sum_{k=1}^{K} d_{\cos}\!\left(\mathbf w,\,\mathbf w_k(\mathbf w)\right)},
\end{equation}
where $\{\mathbf w_k(\mathbf w)\}_{k=1}^{K}$ are the $K$ nearest neighbors of $\mathbf w$ within $\mathcal M_j$, excluding $\mathbf w$ itself. The denominator estimates the local neighborhood scale around $\mathbf w$, calibrating distances before aggregation. Unlike clip-level normalization, which applies a single global scale, sub-band LDN uses band-local scales so that no band systematically dominates the uniform aggregation. We use $K=0$ to denote the raw score without LDN. For $K\geq1$, the resulting clip-level score is the uniform aggregation of normalized band scores:
\begin{equation}
\label{eq:clip-avg}
S^{\mathrm{norm}}_{\mathrm{sub}}(y)
:=
\frac{1}{N_b}\sum_{j=1}^{N_b} d^{\mathrm{norm}}_{\mathrm{sub}}(\mathbf w_{j}(y)).
\end{equation}

\subsection{Dynamic Mean--Max Integration (DMM)}
\label{sec:dmm}
To handle diverse temporal dynamics, \textit{AdaBEAM} computes two parallel clip-level scores: $S_{\mathrm{Tmean}}(y)$ from temporal mean pooling, which captures stable spectral structure, and $S_{\mathrm{Tmax}}(y)$ from temporal max pooling, which emphasizes transient feature activations. We integrate them using a simple parameter-free rule:
\begin{equation}
\label{eq:dmm}
S_{\text{DMM}}(y)=\operatorname{agg}\!\left(S_{\text{Tmean}}(y),\ S_{\text{Tmax}}(y)\right),
\end{equation}
where $\operatorname{agg}$ denotes a fusion operator. By default, we use the mean, and we evaluate alternative choices in our analysis.

Building on the variance-inflation mechanism identified for global matching in Sec.~\ref{sec:global-decomp}, Sec.~\ref{sec:theory} analyzes how sub-band scoring can reduce normal-score variance and establishes an SDT-based sufficient condition for improved detection sensitivity ($d'$).

\section{Theoretical Analysis}
\label{sec:theory}

We analyze \textit{BEAM} at the score level under broad normal operating-condition variability (e.g., machine-state changes and background noise). Proposition~1 upper-bounds the normal-score variance of sub-band scoring by a constant multiple of that of tied-reference global matching. Theorem~1 converts this variance bound into a sufficient condition for improved signal-detection-theory (SDT) sensitivity $d'$, determined by how well the anomalous--normal mean gap is preserved.

Recall that $S_{\mathrm{glob}}$ is the raw cosine-distance score obtained by matching the complete clip representation to one normal reference, whereas $S_{\mathrm{sub}}$ is the raw \textit{BEAM} score obtained by independently matching each sub-band and uniformly aggregating the resulting distances; subscripts $\mathcal N$ and $\mathcal A$ below denote the normal and anomalous classes.

To isolate the effects of aggregation and reference selection, let
$D_{j,i}(y):=d_{\cos}(\mathbf w_j(y),\mathbf w_{i,j})$ denote the
distance between test band $j$ and reference $i$. We then introduce the
tied-reference uniform score
\begin{equation}
\label{eq:Suni_main}
S_{\mathrm{uni}}(y)
:=
\frac{1}{N_b}
\sum_{j=1}^{N_b}D_{j,i^\ast(y)}(y),
\end{equation}
where $i^\ast(y)$ is the globally selected reference.
Thus, $S_{\mathrm{uni}}$ shares the tied reference of
$S_{\mathrm{glob}}$ and the uniform aggregation of $S_{\mathrm{sub}}$.

The band-wise selection gain is then
\begin{equation}
\label{eq:P_main}
P(y)
:=
S_{\mathrm{uni}}(y)-S_{\mathrm{sub}}(y)
\geq 0,
\end{equation}
where nonnegativity follows because independent band-wise minimization
cannot produce a larger distance than using the tied global reference.

For the variance decomposition, let $G$ denote a coarse operating-regime descriptor fixed a priori (e.g., machine state or background-noise condition), which separates within-regime from between-regime normal-score variability. For a disjoint partition, comparing $S_{\mathrm{glob}}$ with $S_{\mathrm{uni}}$ isolates the aggregation change in Eq.~\eqref{eq:glob-decomp-sum}. If an overlapping end-aligned window is appended, the results below remain valid, although the comparison also includes its deterministic coverage effect.

\vspace{0.5em}
\noindent\textbf{Proposition 1.}
\textit{Assume that the normal-score variables have finite second
moments and that the variance terms used as denominators below are
positive:}
\begin{equation}
\begin{aligned}
0<&\;\mathbb E_{\mathcal N}\!\left[
\mathrm{Var}_{\mathcal N}(S_{\mathrm{uni}}\mid G)\right]<\infty,\\
0<&\;\mathbb E_{\mathcal N}\!\left[
\mathrm{Var}_{\mathcal N}(S_{\mathrm{glob}}\mid G)\right]<\infty,\\
0<&\;\mathrm{Var}_{\mathcal N}(S_{\mathrm{glob}})<\infty.
\end{aligned}
\end{equation}
\textit{Define the within-regime aggregation ratio}
\begin{equation}
\label{eq:g0_main}
g_0
:=
\frac{
\mathbb E_{\mathcal N}\!\left[
\mathrm{Var}_{\mathcal N}(S_{\mathrm{glob}}\mid G)
\right]
}{
\mathbb E_{\mathcal N}\!\left[
\mathrm{Var}_{\mathcal N}(S_{\mathrm{uni}}\mid G)
\right]
},
\end{equation}
\textit{where $g_0>1$ means that tied-reference uniform band
aggregation has lower within-regime variance. Define the
between-to-within regime variance ratio}
\begin{equation}
\label{eq:gamma_main}
\gamma
:=
\frac{
\mathrm{Var}_{\mathcal N}\!\left(
\mathbb E_{\mathcal N}[S_{\mathrm{uni}}\mid G]
\right)
}{
\mathbb E_{\mathcal N}\!\left[
\mathrm{Var}_{\mathcal N}(S_{\mathrm{uni}}\mid G)
\right]
}.
\end{equation}
\textit{Larger $\gamma$ indicates stronger regime shifts. Define the
relative normal-class variability of the tied-reference mismatch}
\begin{equation}
\label{eq:p0_main}
p_0
:=
\frac{\mathrm{Var}_{\mathcal N}(P)}
{\mathrm{Var}_{\mathcal N}(S_{\mathrm{glob}})}.
\end{equation}
\textit{Larger $p_0$ indicates greater normal-class variability in $P$;
its net contribution to the variance bound also depends on the covariance
correction defined next. Define}
\begin{equation}
\label{eq:lambda_main}
\lambda
:=
\left(
\frac{
-2\mathrm{Cov}_{\mathcal N}(S_{\mathrm{sub}},P)
}{
\mathrm{Var}_{\mathcal N}(S_{\mathrm{glob}})
}
\right)_{+},
\qquad
(x)_{+}:=\max(x,0).
\end{equation}
\textit{Here, $\lambda>0$ only when $S_{\mathrm{sub}}$ and $P$ are
negatively correlated. Thus, $\lambda-p_0$ captures the net contribution
associated with replacing tied-reference matching by band-aligned
retrieval. Let}
\begin{equation}
\label{eq:Cvar_def}
C_{\mathrm{var}}
:=
\frac{1+\gamma}{g_0}+\lambda-p_0.
\end{equation}
\textit{Then $C_{\mathrm{var}}\geq0$ and}
\begin{equation}
\label{eq:prop1_main}
\mathrm{Var}_{\mathcal N}(S_{\mathrm{sub}})
\leq
C_{\mathrm{var}}\,
\mathrm{Var}_{\mathcal N}(S_{\mathrm{glob}}).
\end{equation}
Thus, $0<C_{\mathrm{var}}<1$ certifies a reduction in normal-score variance, with $C_{\mathrm{var}}$ capturing the combined effects of band-aligned retrieval and uniform band aggregation. The full proof is provided in Appendix~\ref{app:variance}, where the identity $S_{\mathrm{uni}}=S_{\mathrm{sub}}+P$ is combined with the law of total variance over $G$.

For any anomaly score $S$ with finite class-conditional means and
$0<\mathrm{Var}_{\mathcal N}(S)<\infty$, define the anomalous--normal
mean gap and normal-referenced SDT sensitivity index as
\begin{equation}
\label{eq:dprime_main}
\Delta(S)
:=
\mathbb E_{\mathcal A}[S]-\mathbb E_{\mathcal N}[S],
\qquad
d'(S)
:=
\frac{\Delta(S)}
{\sqrt{\mathrm{Var}_{\mathcal N}(S)}}.
\end{equation}

Using the variance-ratio bound in Eq.~\eqref{eq:prop1_main}, we now state
a sufficient condition for improved detection sensitivity.

\vspace{0.5em}
\noindent\textbf{Theorem 1.}
\textit{Under the assumptions of Proposition~1, additionally assume
$0<\mathrm{Var}_{\mathcal N}(S_{\mathrm{sub}})<\infty$ and
$\Delta(S_{\mathrm{glob}})>0$. If}
\begin{equation}
\label{eq:theorem_condition_main}
\Delta(S_{\mathrm{sub}})
\geq
\sqrt{C_{\mathrm{var}}}\,
\Delta(S_{\mathrm{glob}}),
\end{equation}
\textit{then}
\begin{equation}
\label{eq:theorem_result_main}
d'(S_{\mathrm{sub}})
\geq
d'(S_{\mathrm{glob}}).
\end{equation}

Thus, $d'$ improves when the mean-gap ratio is at least $\sqrt{C_{\mathrm{var}}}$, the upper bound on the normal standard-deviation ratio from Proposition~1. The full proof is provided in Appendix~\ref{app:dprime_proof}.

Proposition~1 and Theorem~1 concern raw scores computed from a fixed pooling view. DMM is a complementary fusion step applied to the resulting \textit{Tmean} and \textit{Tmax} scores: mean fusion can reduce normal-score fluctuations when the two views have comparable scales and are not perfectly correlated, whereas max fusion can preserve sharp intermittent anomalies. LDN is likewise treated as an orthogonal calibration layer. We evaluate both extensions empirically in Sec.~\ref{sec:results}.

%% file: experimental_details.tex
\subsection{Data}
We use five DCASE Task~2 benchmarks for the main evaluations, following the unified protocol in \cite{wilkinghoff2026temporal}: DCASE2020T2 \cite{koizumi_description_2020} for single-domain performance and DCASE2022T2--DCASE2025T2 \cite{dohi_description_2022,dohi_description_2023,Nishida_arXiv2024_01,Nishida_arXiv2025_01} for robustness under source--target domain shifts (Table~\ref{tab:datasets}). These benchmarks provide normal reference data and test sets containing normal and anomalous clips. We do not separately evaluate the multi-sensor IMAD-DS benchmark \cite{Albertini2024}; our scope is the single-channel audio provided through DCASE Task~2.

\begin{table}[!t]
\centering
\caption{Per-section statistics of the main DCASE Task~2 benchmarks, following \cite{wilkinghoff2026temporal}. DCASE2020 has a single domain, shown under source for consistency.}
\label{tab:datasets}
\resizebox{\columnwidth}{!}{%
\begin{tabular}{lccccccc}
\toprule
\multicolumn{1}{c}{Dataset} & \multicolumn{1}{c}{\# Types} & \multicolumn{1}{c}{\# Types} & \multicolumn{1}{c}{Split} & \multicolumn{2}{c}{Source domain} & \multicolumn{2}{c}{Target domain} \\
\cmidrule(lr){5-6} \cmidrule(lr){7-8}
& (total) & (dev/eval) & & Normal & Anomalous & Normal & Anomalous \\
\midrule
DCASE2020T2 \cite{koizumi_description_2020} & \multirow{2}{*}{6} & \multirow{2}{*}{6/6} & Train & $\leq 1000$ & - & - & - \\
                                            & & & Test  & $\leq 400$  & $\leq 200$ & - & - \\
\midrule
DCASE2022T2 \cite{dohi_description_2022} & \multirow{2}{*}{7} & \multirow{2}{*}{7/7} & Train & 990 & - & 10 & - \\
                                         & & & Test  & 50  & 50 & 50 & 50 \\
\midrule
DCASE2023T2 \cite{dohi_description_2023} & \multirow{2}{*}{14} & \multirow{2}{*}{7/7} & Train & 990 & - & 10 & - \\
                                         & & & Test  & 50  & 50 & 50 & 50 \\
\midrule
DCASE2024T2 \cite{Nishida_arXiv2024_01} & \multirow{2}{*}{16} & \multirow{2}{*}{7/9} & Train & 990 & - & 10 & - \\
                                         & & & Test  & 50  & 50 & 50 & 50 \\
\midrule
DCASE2025T2 \cite{Nishida_arXiv2025_01} & \multirow{2}{*}{14} & \multirow{2}{*}{7/7} & Train & 990 & - & 10 & - \\
                                         & & & Test  & 50  & 50 & 50 & 50 \\
\bottomrule
\end{tabular}
}
\vspace{-2mm}
\end{table}

\noindent \textbf{DCASE2020T2.}
The DCASE2020T2 benchmark \cite{koizumi_description_2020} focuses on the evaluation of individual machine types, including ToyCar and ToyConveyor from the ToyADMOS \cite{Koizumi_WASPAA2019_01} dataset, and Fan, Pump, Slider, and Valve from the MIMII dataset \cite{purohit_mimii_2019}. It is divided into two sections, the DCASE2020T2 Development Set (sec-dev) and the DCASE2020T2 Evaluation Set (sec-eval). While both sections cover the same machine types, they differ in their machine sections and therefore exhibit distinct acoustic characteristics. All recordings are single-channel at 16 kHz.

\noindent \textbf{DCASE2022T2--DCASE2025T2.}
These benchmarks \cite{dohi_description_2022,dohi_description_2023,Nishida_arXiv2024_01,Nishida_arXiv2025_01} evaluate domain generalization under source--target distribution shifts. As summarized in Table~\ref{tab:datasets}, each section provides predominantly source-domain normal reference data with only a small number of target-domain normal samples, while the test data contain balanced normal and anomalous samples from both domains. DCASE2022T2 uses recordings from MIMII-DG \cite{Dohi2022} and ToyADMOS2 \cite{Harada2021}; sec-dev and sec-eval cover the same seven machine types, with three sections per type in each split and 10~s recordings. DCASE2023T2 and DCASE2024T2 instead use disjoint sec-dev and sec-eval machine types, with seven types in sec-dev and seven and nine types in sec-eval, respectively; each type corresponds to one section, and recordings are 6--18~s long (typically 10--12~s) \cite{haradatoyadmos2+,Albertini2024}. DCASE2025T2 consists of MIMII-DG recordings, ToyADMOS2025, and additional audio collected under the same operating conditions as IMAD-DS \cite{Albertini2024}; its splits contain seven disjoint types each, with one section per type and 6--10~s recordings. Thus, from DCASE2023T2 onward, sec-dev and sec-eval contain non-overlapping machine types, so parameters selected on sec-dev are evaluated under cross-machine-type transfer and may not generalize directly to sec-eval.

\subsection{Metrics}
We report performance using the area under the receiver operating characteristic (ROC) curve (AUC) and the partial AUC (pAUC) over a low false positive rate (FPR) range $[0,p]$, following \cite{koizumi_description_2020}. The ROC curve plots the true positive rate (TPR) against the FPR as the anomaly-score threshold is varied. AUC can be interpreted as the probability that a randomly chosen anomalous sample receives a higher anomaly score than a randomly chosen normal sample. pAUC computes the same ranking-based quantity but restricts attention to the ROC segment with $\mathrm{FPR}\le p$, emphasizing performance under strict false-alarm constraints. We use $p=0.1$ \cite{koizumi_description_2020}.

For DCASE2020T2 (sec-dev and sec-eval), we follow \cite{local_den_norm_kevin} and report the arithmetic mean of AUC and pAUC, where each metric is first averaged across machine sections. For the domain-shift benchmarks DCASE2022T2--DCASE2025T2 (sec-dev and sec-eval), we follow the corresponding Official Score definitions \cite{dohi_description_2022,dohi_description_2023,Nishida_arXiv2024_01,Nishida_arXiv2025_01}: for each machine type, we compute the harmonic mean of (i) source-domain AUC, (ii) target-domain AUC, and (iii) mixed-domain pAUC. The final Official Score is the harmonic mean across machine types.

We also report the detection sensitivity index $d'(S)=\Delta(S)/\sqrt{\mathrm{Var}(S\mid\mathcal N)}$, which measures score separation between anomalous and normal samples relative to the variability of normal scores; larger $d'$ indicates easier detection. See Eq.~\eqref{eq:dprime_main} for the full SDT definition.

\subsection{Implementation Details}
\label{sub-sec:impl}
\subsubsection{Baseline Training-free Frontend System} 
We use both handcrafted and deep spectrogram features. Log-Mel features use 128 mel bins over 0--8000~Hz with frame length 64 and frame shift 32. MFCCs use 90 coefficients with the same frequency range and frame parameters. For LPC, we extract 60 LPC coefficients using Librosa \cite{brian_mcfee-proc-scipy-2015}; for LPC spectrum, we convert these 60 coefficients into an 8000-bin frequency response up to Nyquist. Each audio segment is zero-padded or truncated to 10~s. For deep features, we use frozen embeddings from the last Transformer layer of BEATs iter3 \cite{chen_beats_2022}, EAT \cite{ijcai2024p421} (epoch~30), and EAT-large (epoch~20). The model sizes are approximately 90\,M parameters (BEATs iter3), 85\,M (EAT), and 302\,M (EAT-large). For BEATs iter3 and EAT, input log-mel spectrograms are computed over 20--8000~Hz with frame length 25 and frame shift 10, then standardized using AudioSet statistics \cite{10887974}.

For Log-Mel, MFCC, BEATs, EAT, and EAT-large features, we use Temporal Mean Pooling (\emph{Tmean}) over time, yielding 128-, 90-, 6144-, 6144-, and 8192-dimensional clip-level vectors, respectively. For EAT, we do not use the CLS embedding in order to preserve full band structure, unlike \cite{wilkinghoff2025density}.

\begin{table}[!t]
\centering
\caption{\textit{BEAM} windowing parameters ($C$: window dimension; $N_b$: number of windows). ``Main'' gives benchmark-specific settings, while ``Across-dev'' gives feature-specific settings selected across the five sec-dev benchmarks.}
\label{tab:winspec_params}
\resizebox{\columnwidth}{!}{%
\begin{tabular}{lccccc}
\toprule
& & \multicolumn{2}{c}{Main} & \multicolumn{2}{c}{Across-dev} \\
\cmidrule(lr){3-4}\cmidrule(lr){5-6}
Feature & Benchmark(s) & $C$ & $N_b$ & $C$ & $N_b$ \\
\midrule
\multirow{4}{*}{Log-Mel}
 & 2020/2023 & 38 & 4 & \multirow{4}{*}{38} & \multirow{4}{*}{4} \\
 & 2022 & 25 & 6 & & \\
 & 2024 & 76 & 2 & & \\
 & 2025 & 12 & 11 & & \\
\addlinespace[1pt]
\multirow{2}{*}{MFCC}
 & 2020/2023/2024 & 20 & 5 & \multirow{2}{*}{20} & \multirow{2}{*}{5} \\
 & 2022/2025 & 10 & 9 & & \\
\addlinespace[1pt]
\multirow{3}{*}{LPC spectrum}
 & 2020/2023/2024 & 3200 & 3 & \multirow{3}{*}{3200} & \multirow{3}{*}{3} \\
 & 2022 & 2400 & 4 & & \\
 & 2025 & 1600 & 5 & & \\
\addlinespace[1pt]
BEATs/EAT & 2020, 2022--2025 & 768 & 8 & 768 & 8 \\
EAT-large & 2020, 2022--2025 & 1024 & 8 & 1024 & 8 \\
BEATs-FT & 2020/2023/2024 & 768 & 8 & 768 & 8 \\
\bottomrule
\end{tabular}
}
\vspace{-5mm}
\end{table}

\subsubsection{Proposed Training-Free Frontend System}
Except for LPC spectrum, \textit{BEAM} uses the \emph{Tmean} representation, while \textit{AdaBEAM} mean-fuses the \emph{Tmean} and \emph{Tmax} scores. Table~\ref{tab:winspec_params} reports two sets of numerical $(C,N_b)$ values. The Main columns retain the benchmark-specific settings selected through a coarse sweep on the corresponding sec-dev set and fixed for its sec-eval set without using sec-eval information; these settings are used in Tables~\ref{tab:main_without_label}, \ref{tab:main_with_label}, and~\ref{tab:scoring_ablation_avg}. The Across-dev columns give one setting per feature selected by aggregating performance across the five sec-dev benchmarks; these fixed settings are used for the neighborhood-size analysis in Fig.~\ref{fig:knorm_comparison}. Figs.~\ref{fig:theory_main_results} and~\ref{fig:window_size_comparison} instead sweep $C$ and therefore do not use a single selected width. Moreover, leave-one-year-out selection achieved comparable sec-eval gains to year-specific selection (Appendix~Table~\ref{tab:loyo_band_selection}), indicating that the improvements do not depend strongly on benchmark-specific tuning of $C$. Fine-tuned BEATs uses the same width as frozen BEATs.

We use stride $s=C$ throughout. For handcrafted features, if the regular length-$C$ windows do not cover all $F$ coordinates, we append one end-aligned window $[F-C,F)$, which may overlap with the preceding window.

\input{results/main_training_free}

\subsubsection{Additional Finetuned Frontend System}
Under the same settings, we evaluate \textit{BEAM} and \textit{AdaBEAM} on fine-tuned BEATs encoders (SSL checkpoint iter3), adapting them to the target distribution using meta-information available from the reference normal data, following the AnoPatch recipe~\cite{jiang_anopatch_2024}. We perform full fine-tuning (updating all encoder parameters) with the AnoPatch outlier-exposure classification objective guided by this meta-information. To ensure reproducibility, we only instantiate the implementation choices that are not fully specified in~\cite{jiang_anopatch_2024} and do not tune them beyond the official development split: (i) audio preprocessing, where each segment is zero-padded or truncated to 10\,s without standardization, enhancement, or segmentation; (ii) an attentive statistics pooling layer producing a 768-dimensional representation; and (iii) an ArcFace loss~\cite{arcface_8953658} with margin $0.5$ and scale $64$. All remaining hyperparameters follow~\cite{jiang_anopatch_2024} (or the AnoPatch public configuration, when available), and we perform no additional tuning.

Models are trained for 10k steps with a batch size of 32, gradient accumulation of 8, and the AdamW optimizer~\cite{Loshchilov2017DecoupledWD} with a learning rate of $1\times10^{-4}$ and a linear scheduler (960 warm-up steps). Following the AnoPatch pipeline~\cite{jiang_anopatch_2024}, we use SpecAugment~\cite{Park2019SpecAugmentAS} with a time mask of 80 and a frequency mask of 40 for DCASE2020T2, and a frequency mask of 80 for DCASE2023T2 and DCASE2024T2. Similar to~\cite{wilkinghoff2025density, fujimura2025asdkit}, classification targets correspond to machine sections in DCASE2020T2 and to machine attributes in DCASE2023T2 and DCASE2024T2, with 41, 167, and 141 classes, respectively. For simplicity and consistency, we use the same preprocessing as in the training-free frontend system.

We do not tune training hyperparameters, and instead follow the main settings from~\cite{jiang_anopatch_2024} for fair comparison. Nevertheless, frozen BEATs outperform fine-tuning on DCASE2023T2 and DCASE2024T2, consistent with~\cite{wilkinghoff2025density}. Accordingly, we extract embeddings from intermediate layers (layer~6 and layer~8, respectively) to mitigate pretraining--downstream mismatch~\cite{10887974}. Because sec-eval contains machine types absent from sec-dev in these DG benchmarks, we follow the sec-dev selection protocol in~\cite{wilkinghoff2025density}, yielding a first-shot setting where tuning uses disjoint machine types. For DCASE2020T2, we use the last layer by default.

\subsubsection{Anomaly Detector.} 
Following \cite{wilkinghoff2025density}, we use cosine distance and $K=1$ as the LDN neighborhood size in our main experiments. Global matching estimates the local scale around the selected full-clip reference using its neighbors in the global memory. \textit{BEAM} and \textit{AdaBEAM} apply the same normalization independently to each sub-band, using the neighborhood of the selected band reference within its corresponding memory $\mathcal M_j$. We explore the effect of varying $K$ in Sec.~\ref{sec:detectors} and provide a controlled comparison with Euclidean distance in Appendix~\ref{app:distance_metric}.

%% file: results/main_training_free.tex
\definecolor{sectionrow}{gray}{0.93} 
\providecommand{\gainup}[1]{{\scriptsize\textcolor{green!45!black}{$\uparrow$#1}}}
\providecommand{\gaindown}[1]{{\scriptsize\textcolor{red!70!black}{$\downarrow$#1}}}
\providecommand{\gainflat}[1]{{\scriptsize\textcolor{gray!70!black}{$\rightarrow$#1}}}
\def\stripgainplus+#1\relax{#1}
\def\stripgainminus-#1\relax{#1}
\providecommand{\gain}[1]{%
  \ifdim #1pt>0pt
    \expandafter\gainup\expandafter{\stripgainplus#1\relax}%
  \else\ifdim #1pt<0pt
    \expandafter\gaindown\expandafter{\stripgainminus#1\relax}%
  \else
    \gainflat{0.0}%
  \fi\fi}
\providecommand{\blockbest}[1]{\underline{#1}}
\providecommand{\best}[1]{\textbf{\underline{#1}}}
\begin{table*}[!ht]
\centering
\caption{Comparison of training-free frontends with task-specific systems across five DCASE Task~2 benchmarks. Training-free methods use neither task-specific training nor auxiliary labels. Overall Mean averages the ten sec-dev and sec-eval scores. Up, down, and right arrows denote improvement, degradation, and no change from the matched global baseline. Underlining marks the best result within each training-free feature block, bold underlining the best across those blocks, and bold the best task-specific result.}
\label{tab:main_without_label}
\setlength{\tabcolsep}{2.5pt}
\resizebox{\textwidth}{!}{%
\renewcommand{\arraystretch}{0.96}%
\begin{tabular}{l*{11}{c}}
\toprule & 
\multicolumn{2}{c}{DCASE2020T2} & 
\multicolumn{2}{c}{DCASE2022T2 (DG)} & 
\multicolumn{2}{c}{DCASE2023T2 (DG)} & 
\multicolumn{2}{c}{DCASE2024T2 (DG)} &
\multicolumn{2}{c}{DCASE2025T2 (DG)} & \\
\cmidrule(lr){2-3} \cmidrule(lr){4-5} \cmidrule(lr){6-7}
\cmidrule(lr){8-9} \cmidrule(lr){10-11}
 Method & \shortstack{sec-dev\\{\scriptsize (Amean)}} & \shortstack{sec-eval\\{\scriptsize (Amean)}}
 & \shortstack{sec-dev\\{\scriptsize (Hmean)}} & \shortstack{sec-eval\\{\scriptsize (Hmean)}}
 & \shortstack{sec-dev\\{\scriptsize (Hmean)}} & \shortstack{sec-eval\\{\scriptsize (Hmean)}}
 & \shortstack{sec-dev\\{\scriptsize (Hmean)}} & \shortstack{sec-eval\\{\scriptsize (Hmean)}}
 & \shortstack{sec-dev\\{\scriptsize (Hmean)}} & \shortstack{sec-eval\\{\scriptsize (Hmean)}} & \shortstack{Overall\\Mean} \\
\midrule
\rowcolor{sectionrow}\multicolumn{12}{l}{Task-specific systems trained separately for each machine type}\\
AE \cite{koizumi_description_2020, dohi_description_2023, Nishida_arXiv2024_01} &66.6 &70.0 &- &- &56.9 &61.1 &55.4 &56.5 &- &- &- \\
AE-IDNN \cite{suefusa_anomalous_2020} &71.5 &- &- &- &- &- &- &- &- &- &- \\
ANP-IDNN \cite{wichern_anomalous_2021} &71.8 &75.6 &- &- &- &- &- &- &- &- &- \\
PAE \cite{10095568} &74.2 &- &- &- &- &- &- &- &- &- &- \\
AudDSR \cite{zavrrtanki_10447941} &78.2 &- &- &- &- &- &- &- &- &- &- \\
MemNMF \cite{saengthong2026memnmf} &74.9 &75.6 &- &- &- &- &- &- &- &- &- \\
\midrule
\rowcolor{sectionrow}\multicolumn{12}{l}{Task-specific systems trained across multiple machine types with auxiliary labels} \\
UnlabeledASD \cite{10890020} (PL) &- &- &- &- &63.5 \std{0.4} &65.5 \std{0.5} &- &- &- &- &- \\
SSL4ASD w/ LDN \cite{local_den_norm_kevin, wilkinghoff2025density} &90.7 &90.2 &- &- &68.4 &68.0 &62.0 &54.7 &- &- &- \\
\hspace{0.5em}+ 10 ensembles &\textbf{94.2} &93.3 &- &- &\textbf{71.3} &72.4 &\textbf{65.2} &56.5 &- &- &- \\
AnoPatch \cite{jiang_anopatch_2024} (BEATs Full FT) &90.9 &\textbf{94.3} &- &- &64.2 &74.2 &62.5 \std{0.8} &65.6 \std{1.1} &- &- &- \\
\hspace{0.5em}+ pseudo labeling \cite{10889514} &- &- &- &- &- &- &64.1 \std{0.6} &\textbf{66.0} \std{0.7} &- &- &- \\
AnoPatch-LoRA \cite{anopatch_lora} &- &- &- &- &67.3 &\textbf{77.2} &- &- &- &- &- \\
UnlabeledASD \cite{10890020} &- &- &- &- &67.2 \std{0.7} &68.8 \std{0.9} &- &- &- &- &- \\
ASDKit \cite{fujimura2025asdkit} (BEATs-LoRA) &90.6 \std{0.2} &93.5 \std{0.3} &- &- &62.7 \std{0.7} &71.5 \std{0.5} &59.5 \std{0.9} &61.7 \std{0.7} &- &- &- \\
\midrule
\rowcolor{sectionrow}\multicolumn{12}{l}{Training-free systems (no task-specific training)}\\
\multicolumn{12}{c}{\textbf{\textit{Handcrafted feature: Log-Mel}}} \\
Global matching baseline &73.1 &71.9 &57.1 &58.0 &51.7 &60.5 &57.4 &57.4 &53.8 &57.2 &59.8 \\
\hspace{0.5em}+ DMM &77.4 \gain{+4.3} &74.9 \gain{+3.0} &58.4 \gain{+1.3} &\blockbest{58.8} \gain{+0.8} &53.3 \gain{+1.6} &61.9 \gain{+1.4} &\blockbest{58.5} \gain{+1.1} &58.0 \gain{+0.6} &58.1 \gain{+4.3} &\blockbest{58.0} \gain{+0.8} &61.7 \\
BEAM (proposed) &76.7 \gain{+3.6} &75.4 \gain{+3.5} &59.7 \gain{+2.6} &58.0 \gain{+0.0} &56.3 \gain{+4.6} &67.8 \gain{+7.3} &57.9 \gain{+0.5} &57.9 \gain{+0.5} &57.0 \gain{+3.2} &\blockbest{58.0} \gain{+0.8} &62.5 \\
AdaBEAM (proposed) &\blockbest{80.6} \gain{+7.5} &\blockbest{78.2} \gain{+6.3} &\blockbest{60.9} \gain{+3.8} &58.7 \gain{+0.7} &\blockbest{56.9} \gain{+5.2} &\blockbest{67.9} \gain{+7.4} &\blockbest{58.5} \gain{+1.1} &\blockbest{58.1} \gain{+0.7} &\blockbest{59.6} \gain{+5.8} &57.8 \gain{+0.6} &\blockbest{63.7} \\
\midrule
\multicolumn{12}{c}{\textbf{\textit{Handcrafted feature: MFCC}}} \\
Global matching baseline &74.2 &72.5 &56.4 &57.2 &52.5 &64.3 &57.3 &54.5 &53.0 &57.7 &60.0 \\
\hspace{0.5em}+ DMM &76.6 \gainup{2.4} &76.4 \gainup{3.9} &56.9 \gainup{0.5} &57.7 \gainup{0.5} &53.0 \gainup{0.5} &63.4 \gaindown{0.9} &58.5 \gainup{1.2} &53.8 \gaindown{0.7} &55.1 \gainup{2.1} &57.7 \gainflat{0.0} &60.9 \\
BEAM (proposed) &77.1 \gain{+2.9} &75.4 \gain{+2.9} &\blockbest{61.0} \gain{+4.6} &58.0 \gain{+0.8} &56.8 \gain{+4.3} &\blockbest{64.8} \gain{+0.5} &58.2 \gain{+0.9} &\blockbest{56.6} \gain{+2.1} &55.4 \gain{+2.4} &\blockbest{59.8} \gain{+2.1} &62.3 \\
AdaBEAM (proposed) &\blockbest{77.6} \gain{+3.4} &\blockbest{76.5} \gain{+4.0} &60.8 \gain{+4.4} &\blockbest{59.6} \gain{+2.4} &\blockbest{57.0} \gain{+4.5} &64.6 \gain{+0.3} &\blockbest{58.7} \gain{+1.4} &56.5 \gain{+2.0} &\blockbest{57.1} \gain{+4.1} &59.6 \gain{+1.9} &\blockbest{62.8} \\
\midrule
\multicolumn{12}{c}{\textbf{\textit{Handcrafted feature: LPC Spectrum}}} \\
Global matching baseline &76.8 &76.4 &61.2 &59.4 &58.6 &64.6 &56.2 &55.6 &56.4 &\blockbest{58.5} &62.4 \\
BEAM (proposed) &\blockbest{77.9} \gain{+1.1} &\blockbest{80.0} \gain{+3.6} &\blockbest{64.6} \gain{+3.4} &\blockbest{61.6} \gain{+2.2} &\blockbest{60.6} \gain{+2.0} &\blockbest{66.2} \gain{+1.6} &\blockbest{58.5} \gain{+2.3} &\blockbest{57.3} \gain{+1.7} &\blockbest{59.3} \gain{+2.9} &57.8 \gain{-0.7} &\blockbest{64.4} \\
\midrule
\multicolumn{12}{c}{\textbf{\textit{Deep feature: EAT}}} \\
Global matching baseline &77.8 &77.7 &58.0 &59.5 &61.3 &66.0 &57.1 &58.9 &56.7 &59.1 &63.2 \\
\hspace{0.5em}+ DMM &78.5 \gain{+0.7} &78.6 \gain{+0.9} &58.3 \gain{+0.3} &60.1 \gain{+0.6} &61.6 \gain{+0.3} &66.2 \gain{+0.2} &57.3 \gain{+0.2} &59.1 \gain{+0.2} &57.3 \gain{+0.6} &59.6 \gain{+0.5} &63.7 \\
BEAM (proposed) &84.3 \gain{+6.5} &86.1 \gain{+8.4} &61.8 \gain{+3.8} &64.1 \gain{+4.6} &\blockbest{63.9} \gain{+2.6} &\blockbest{71.7} \gain{+5.7} &\blockbest{60.2} \gain{+3.1} &\blockbest{61.7} \gain{+2.8} &61.0 \gain{+4.3} &60.8 \gain{+1.7} &67.6 \\
AdaBEAM (proposed) &\blockbest{84.4} \gain{+6.6} &\blockbest{86.5} \gain{+8.8} &\blockbest{61.9} \gain{+3.9} &\blockbest{64.3} \gain{+4.8} &\blockbest{63.9} \gain{+2.6} &71.5 \gain{+5.5} &\blockbest{60.2} \gain{+3.1} &\blockbest{61.7} \gain{+2.8} &\blockbest{61.3} \gain{+4.6} &\blockbest{61.1} \gain{+2.0} &\blockbest{67.7} \\
\midrule
\multicolumn{12}{c}{\textbf{\textit{Deep feature: EAT-Large}}} \\
Global matching baseline from \cite{wilkinghoff2025density} &79.3 &79.9 &- &- &\blockbest{65.0} &66.1 &56.8 &58.9 &- &- &- \\
Global matching baseline from \cite{wilkinghoff2026temporal} &79.1 &82.2 &\blockbest{63.2} &62.5 &62.9 &65.6 &58.9 &60.9 &\blockbest{63.0} &58.5 &65.6 \\
Global matching baseline (ours) &78.2 &78.0 &59.4 &59.1 &61.5 &64.7 &56.8 &59.8 &57.6 &59.3 &63.4 \\
\hspace{0.5em}+ DMM &79.1 \gain{+0.9} &79.5 \gain{+1.5} &60.0 \gain{+0.6} &59.9 \gain{+0.8} &61.9 \gain{+0.4} &65.0 \gain{+0.3} &57.3 \gain{+0.5} &60.3 \gain{+0.5} &58.1 \gain{+0.5} &60.0 \gain{+0.7} &64.1 \\
BEAM (proposed) &84.2 \gain{+6.0} &85.3 \gain{+7.3} &62.6 \gain{+3.2} &62.9 \gain{+3.8} &64.0 \gain{+2.5} &\blockbest{71.2} \gain{+6.5} &59.8 \gain{+3.0} &62.6 \gain{+2.8} &61.0 \gain{+3.4} &\blockbest{60.9} \gain{+1.6} &67.5 \\
AdaBEAM (proposed) &\blockbest{84.7} \gain{+6.5} &\blockbest{86.1} \gain{+8.1} &63.0 \gain{+3.6} &\blockbest{63.4} \gain{+4.3} &64.2 \gain{+2.7} &71.1 \gain{+6.4} &\blockbest{59.9} \gain{+3.1} &\best{63.0} \gain{+3.2} &61.7 \gain{+4.1} &\blockbest{60.9} \gain{+1.6} &\blockbest{67.8} \\
\midrule
\multicolumn{12}{c}{\textbf{\textit{Deep feature: BEATs-iter3}}} \\
Global matching baseline from \cite{10887974} &74.7 &- &- &- &- &\best{73.8} &- &- &- &- &- \\
Global matching baseline from \cite{wilkinghoff2025density} &81.5 &82.2 &- &- &64.8 &67.6 &58.1 &62.4 &- &- &- \\
Global matching baseline from \cite{wilkinghoff2026temporal} &86.8 &87.4 &\best{66.1} &\best{65.1} &65.7 &70.8 &58.7 &60.1 &\best{64.0} &\best{62.4} &68.7 \\
Global matching baseline (ours) &82.1 &82.7 &60.0 &61.6 &63.8 &70.8 &57.7 &60.0 &57.3 &60.4 &65.6 \\
\hspace{0.5em}+ DMM &82.8 \gain{+0.7} &83.9 \gain{+1.2} &60.5 \gain{+0.5} &61.9 \gain{+0.3} &64.0 \gain{+0.2} &70.9 \gain{+0.1} &57.8 \gain{+0.3} &60.2 \gain{+0.2} &58.2 \gain{+0.9} &60.3 \gain{-0.1} &66.1 \\
BEAM (proposed) &86.9 \gain{+4.8} &87.8 \gain{+5.1} &63.6 \gain{+3.6} &64.2 \gain{+2.6} &66.1 \gain{+2.3} &73.1 \gain{+2.3} &62.2 \gain{+4.5} &\blockbest{62.6} \gain{+2.6} &62.5 \gain{+5.2} &62.3 \gain{+1.9} &69.1 \\
AdaBEAM (proposed) &\best{87.6} \gain{+5.5} &\best{88.5} \gain{+5.8} &64.0 \gain{+4.0} &64.7 \gain{+3.1} &\best{66.3} \gain{+2.5} &73.1 \gain{+2.3} &\best{62.4} \gain{+4.7} &\blockbest{62.6} \gain{+2.6} &63.0 \gain{+5.7} &62.3 \gain{+1.9} &\best{69.5} \\
\bottomrule

\end{tabular}
}
\end{table*}

%% file: single_experiments.tex
Section~\ref{sec:dcase_benchmarks} presents the main training-free benchmark results on DCASE Task~2 (2020/2022/2023/2024/2025) under the implementation described in Section~\ref{sub-sec:impl}, links the gains to variance reduction and sensitivity, and reports fine-tuned encoder performance as a secondary reference. Section~\ref{sec:detectors} then provides diagnostics on when the improvements hold, including the effects of sub-band window size, normalization neighborhood size $K$, and the roles of \textit{BEAM} and \textit{AdaBEAM}. Section~\ref{sec:computational_cost} reports inference cost, while Appendix~Sec.~\ref{app:mimii_analysis} provides a complementary MIMII study of clip-level feature behavior under varying noise conditions.

\begin{figure*}[!ht]
    \centering
    \includegraphics[width=0.99\textwidth]{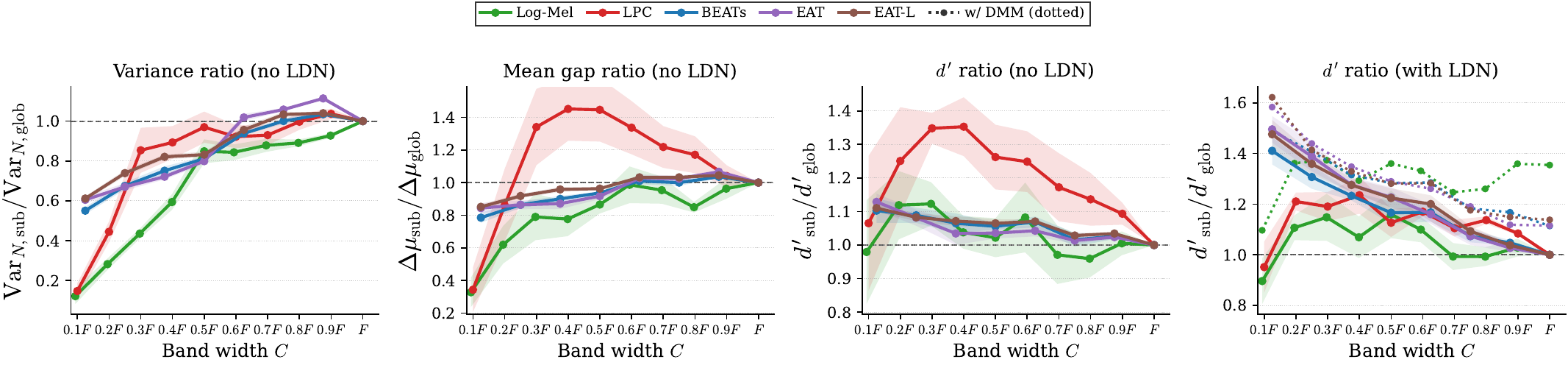}
    \caption{Cross-feature empirical diagnostics for the SDT analysis on DCASE2020T2, DCASE2022T2, DCASE2023T2, DCASE2024T2, and DCASE2025T2. For each benchmark, each ratio is first averaged over sec-dev and sec-eval; curves and shaded regions then show the mean $\pm$ standard error of the mean (SEM) across the five benchmark-level averages. From left to right: normal-score variance ratio, anomalous--normal mean-gap ratio, $d'$ ratio without LDN, and $d'$ ratio with LDN. Solid curves show \textit{BEAM}; dotted curves in the right panel additionally apply mean DMM fusion. For the solid curves, $C=F$ recovers global Tmean matching and each ratio equals one; the dotted DMM ratios need not equal one because temporal-view fusion remains active. MFCC is omitted because its DCT mixes information across Mel bands, so contiguous coefficient groups are not frequency-local sub-bands.}
    \label{fig:theory_main_results}
    \vspace{-2mm}
\end{figure*}

\input{results/main_finetuning}

\begin{figure}[!t]
    \centering
    \includegraphics[width=0.92\columnwidth]{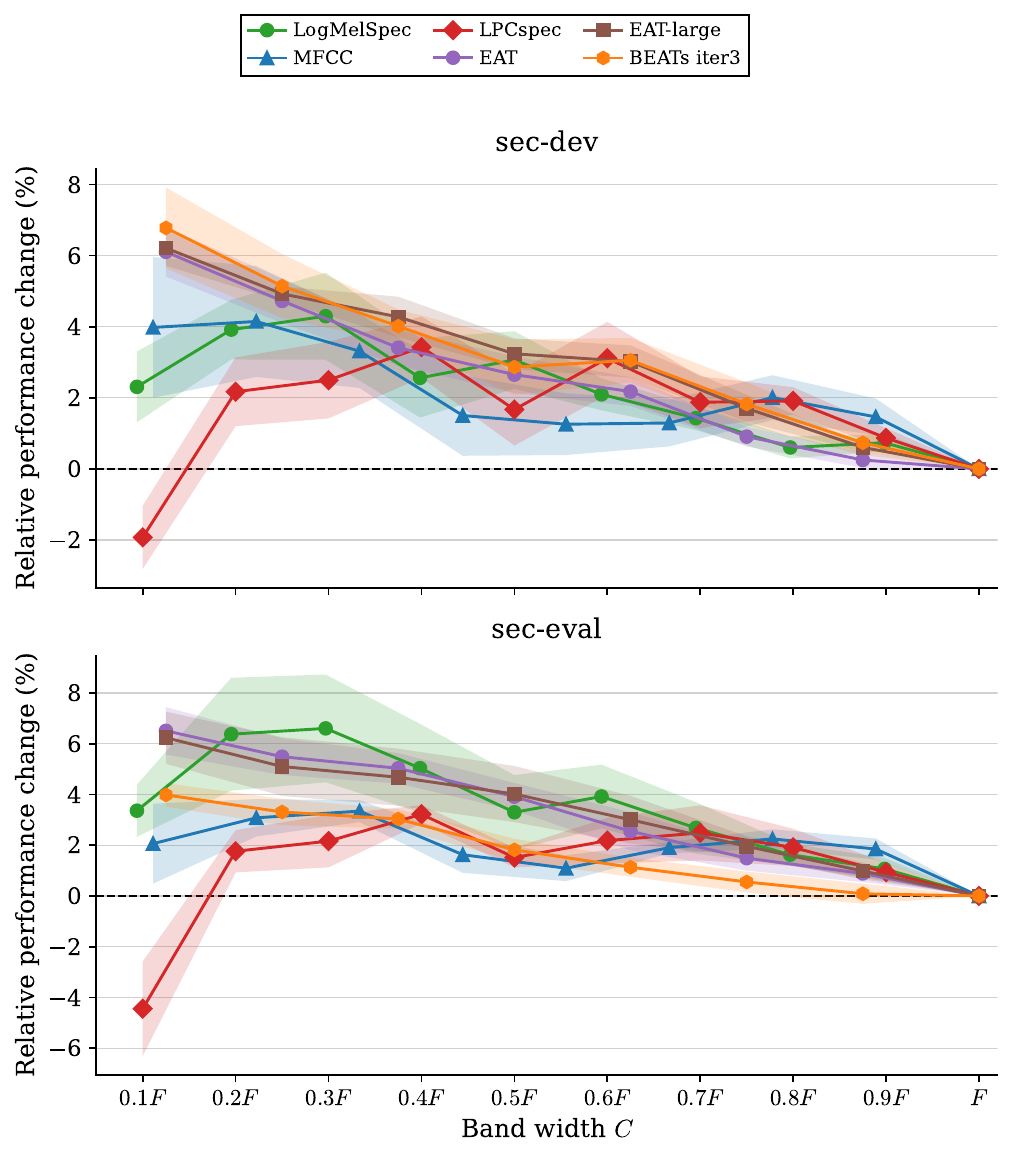}
\caption{Sensitivity of \textit{BEAM} to the sub-band width $C$, shown as the relative performance change from tied-reference global matching and aggregated across DCASE Task~2 benchmarks. Lines show the mean change across datasets, and shaded regions show cross-benchmark variability, separately for sec-dev and sec-eval. The endpoint $C=F$ uses one full-width band and therefore exactly recovers global matching. Deep-feature windows are restricted to widths aligned with their frequency-patch grids.}
    \label{fig:window_size_comparison}
\vspace{-2mm}
\end{figure}

\begin{figure*}[!t]
    \centering
    \includegraphics[width=0.99\textwidth]{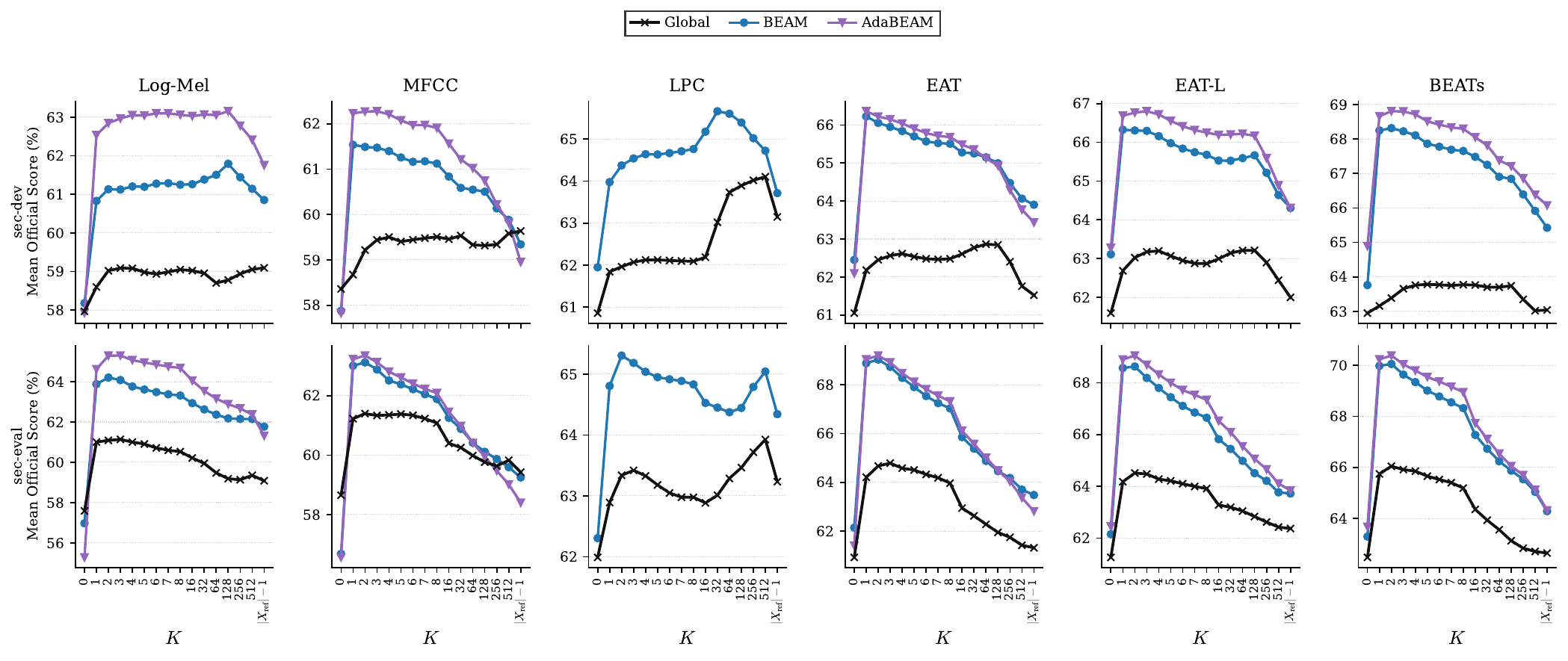}
    \caption{Effect of the LDN neighborhood size $K$ on mean scores across the five DCASE benchmarks. Columns correspond to feature types, with sec-dev on the top row and sec-eval on the bottom row. Curves compare global Tmean matching, \textit{BEAM} with Tmean, and \textit{AdaBEAM} with mean DMM fusion; \textit{AdaBEAM} is not applicable to LPC spectrum. Each feature uses the fixed width $C$ selected across the five sec-dev benchmarks (the Across-dev columns of Table~\ref{tab:winspec_params}); only $K$ is varied. Global $K$ counts neighbors of the selected full-clip reference, whereas sub-band $K$ counts neighbors within the corresponding band-aligned memory $\mathcal M_j$. Since $|\mathcal M_j|=|\mathcal X_{\mathrm{ref}}|$, $K=|\mathcal X_{\mathrm{ref}}|-1$ uses all other references, and $K=0$ denotes scoring without LDN.}
    \label{fig:knorm_comparison}
    \vspace{-2mm}
\end{figure*}

\subsection{Main Results and Findings}
\label{sec:dcase_benchmarks}
Table~\ref{tab:main_without_label} compares our training-free systems, which use neither task-specific training nor auxiliary labels, with previously reported training-free and task-specifically trained systems. Because published results may differ in preprocessing and encoder-layer selection, we assess the proposed scoring methods primarily against our matched global baselines. \textit{BEAM} increased the Overall Mean for every frontend, showing that band-aligned scoring transfers across handcrafted descriptors and pre-trained patch embeddings, as well as across the single-domain and domain-generalization benchmarks. Its largest average gains were obtained with EAT and EAT-large, while \textit{AdaBEAM} with BEATs achieved the highest training-free Overall Mean of 69.5, compared with 65.6 for its matched global baseline. DMM improved the Overall Mean for every applicable frontend, indicating that temporal-view fusion was complementary at the aggregate level. However, it slightly decreased MFCC performance on DCASE2023T2 and DCASE2024T2 sec-eval and BEATs performance on DCASE2025T2 sec-eval, showing that Tmean alone can be preferable when the Tmax view is not complementary. Handcrafted LPC improved from 62.4 to 64.4 overall and surpassed several autoencoder baselines on individual benchmarks, but decreased by 0.7 points on DCASE2025T2 sec-eval. Because its width was selected on sec-dev and the two splits contain different machine types, this isolated decrease may reflect imperfect transfer of the selected band size. Some published baselines remained stronger on particular splits, reflecting differences in their configurations; for example, the domain-wise normalization baseline~\cite{10887974} used machine-type-specific layer selection on DCASE2023T2 sec-eval. Overall, the matched comparisons show that the gains are not tied to a particular representation family or machine-specific tuning, although their magnitude remains feature- and dataset-dependent.

Figure~\ref{fig:theory_main_results} empirically supports the SDT mechanism in Sec.~\ref{sec:theory}. MFCC is excluded from this diagnostic because its cepstral coefficients are globally derived from the Mel spectrum and therefore do not form frequency-local sub-bands; it remains included in the benchmark comparisons in Table~\ref{tab:main_without_label}. Narrower bands reduced normal-score variance most strongly for Log-Mel and LPC, with more moderate reductions for the deep embeddings. However, their effect on anomalous--normal separation differed: LPC increased the mean gap at moderate widths and therefore produced the strongest raw $d'$ gains, whereas Log-Mel sometimes lost too much separation at the finest widths. Accordingly, the handcrafted features favored representation-dependent intermediate widths, whereas deep embeddings largely preserved the mean gap and achieved their largest sensitivity gains at the finest available patch-aligned partition. With LDN, the gains extended over a broader range of widths, particularly for deep embeddings, and DMM often increased them further. The shaded regions indicate that sensitivity gains also depend on the dataset, alongside the representation and sub-band width, consistent with the variance--mean-gap trade-off characterized in Theorem~1.

As a secondary reference, Table~\ref{tab:main_with_label} evaluates BEATs with task-specific fine-tuning using auxiliary labels. Relative to training-free BEATs in Table~\ref{tab:main_without_label}, fine-tuning produced large gains on DCASE2020T2 but mixed results on the domain-generalization benchmarks. Under the default LDN setting, \textit{BEAM} and \textit{AdaBEAM} improved over the matched fine-tuned global baseline on every reported split, whereas DMM alone was less consistent. Fine-tuned \textit{AdaBEAM} achieved the best DCASE2020T2 scores in the table, 95.6 on sec-dev and 96.2 on sec-eval, exceeding SSL4ASD~\cite{wilkinghoff2025density} with ten ensembles and AnoPatch~\cite{jiang_anopatch_2024}. On DCASE2023T2 and DCASE2024T2, however, the strongest prior label-trained systems remained ahead on several splits. Notably, the training-free BEATs+\textit{AdaBEAM} system still outperformed the BEATs-LoRA baseline in ASDKit~\cite{fujimura2025asdkit} on both benchmarks. Domain-wise standardized scoring~\cite{10887974} substantially recovered sec-eval performance on the domain-generalization datasets, particularly DCASE2024T2, but was not uniformly superior and requires domain labels to partition the normal reference data. Overall, the results show that band-aligned scoring remains effective after fine-tuning, while the effects of task-specific training and normalization remain dataset-dependent.


\subsection{Design Analysis}
\label{sec:detectors}

\input{results/table_vi_scoring_ablation_official_dev_eval}


\noindent \textbf{Analysis of Sub-Band Window Sizes.}
Figure~\ref{fig:window_size_comparison} analyzes the effect of sub-band granularity in \textit{BEAM} after aggregating the relative performance change across datasets. On average, most front ends improved over tied-reference global matching throughout a broad range of non-global widths, indicating that the gains did not require a narrowly tuned value of $C$. Deep features benefited most from their finest patch-aligned widths and gradually approached the global baseline as $C$ increased. Handcrafted features generally favored small-to-moderate widths, whereas the finest tested LPC partition degraded performance, showing that overly fine bands can remove useful spectral context. At $C=F$, the score exactly recovers global matching and the relative change is therefore zero. We do not evaluate $C=1$ because one-dimensional cosine distance is degenerate, reducing to sign agreement rather than measuring spectral magnitude differences; the sweep therefore begins at the finest meaningful width supported by each representation. Leave-one-year-out selection of $C$ retained the average gains without access to the target year's development set (Appendix~Table~\ref{tab:loyo_band_selection}).

\noindent \textbf{Impact of $K$ for Sub-Band Density Scoring.} 
Figure~\ref{fig:knorm_comparison} shows how the neighborhood size $K$ affects global matching, \textit{BEAM}, and \textit{AdaBEAM} for each feature type after averaging across all five DCASE benchmarks. Global LDN estimates the local scale around the selected full-clip reference from its $K$ neighbors, whereas \textit{BEAM} and \textit{AdaBEAM} estimate a separate local scale for each selected sub-band reference using $K$ neighbors from the corresponding band-aligned memory $\mathcal M_j$. Since each $\mathcal M_j$ contains one descriptor from every reference sample, $|\mathcal M_j|=|\mathcal X_{\mathrm{ref}}|$, and the maximum neighborhood $K=|\mathcal X_{\mathrm{ref}}|-1$ is numerically the same in both settings. The marked improvement from $K=0$ (without LDN) to a small positive $K$, particularly for \textit{BEAM}, shows that local-scale normalization is important for consistently realizing the benefit of sub-band matching. Small neighborhoods, typically $K=1$--4, produced the strongest average gains for most front ends, while LPC on sec-dev favored a larger neighborhood. The fixed default $K=1$ nevertheless performed well across front ends and splits. Increasing $K$ gradually reduced the advantage of sub-band scoring, particularly for deep features, consistent with findings that neighborhood expansion can cross cluster boundaries and include samples outside the local region, thereby degrading the local density estimate~\cite{wilkinghoff2026mind}. \textit{BEAM} remained above global matching over a broad range of $K$, while \textit{AdaBEAM} usually matched or improved upon \textit{BEAM} where temporal mean--max fusion was applicable. Overall, the averaged trends suggest that a fixed small $K$ transfers well across benchmarks.

\noindent \textbf{Effects of Dynamic Mean–Max Scoring.}
DMM fused scores from temporal mean and temporal max memory banks to combine stable, long-term statistics emphasized by Tmean with transient activations emphasized by Tmax. In Table~\ref{tab:main_without_label}, DMM improved the Overall Mean over Tmean for every applicable k-NN+LDN frontend, although individual splits did not always improve. The cross-dataset averages in Table~\ref{tab:scoring_ablation_avg} show that Mean and Max were the strongest DMM variants on average under a fixed fusion rule. When combined with sub-band matching (\textit{AdaBEAM}), DMM typically provided an additional improvement over \textit{BEAM} using Tmean alone, indicating that view fusion remained beneficial after sub-band reference selection.

\noindent \textbf{Sub-Band Scoring Strategy.}
Table~\ref{tab:scoring_ablation_avg} compares sub-band score aggregation rules for \textit{BEAM}-Tmean across all front ends. Among the evaluated rules, mean aggregation consistently performed best, indicating that integrating evidence from all regions was more reliable than selecting a single extreme. Maximum aggregation was weaker, while minimum aggregation caused the largest drop by focusing on the easiest-matching sub-band and potentially ignoring anomalous regions.



\subsection{Computational Cost}
\label{sec:computational_cost}

\begin{table}[!t]
\centering
\caption{Computational cost per 10\,s clip (mean $\pm$ std over $n_{\mathrm{runs}}=5$). Feature-extraction time measures waveform$\rightarrow$f, and end-to-end (E2E) time measures waveform$\rightarrow$anomaly score, including $1$-NN scoring against a memory bank of $500$ normal clips. Handcrafted frontends are computed on CPU only in this benchmark. CPU: AMD Ryzen 5 5600X. GPU: NVIDIA GeForce RTX 3090. For BEATs, the encoder complexity is reported as MACs/Params (G/M).}
\label{tab:cost_frontends}
\resizebox{\columnwidth}{!}{%
\begin{tabular}{l r r r c}
\toprule
Frontend & Output dim & Feat. (ms) & E2E (ms) & MACs/Params (G/M) \\
\midrule
Log-Mel (Tmean)                  & 128  & $3.95 \pm 0.21$  & $6.20 \pm 0.25$   & N/A \\
MFCC (Tmean)                     & 90   & $3.95 \pm 0.01$  & $7.95 \pm 1.08$   & N/A \\
LPC spectrum                     & 8000 & $31.34 \pm 0.05$ & $41.11 \pm 0.15$  & N/A \\
BEATs iter3 (Tmean, CPU)         & 6144 & $271.15 \pm 6.06$& $494.02 \pm 6.19$ & $45.01/90.71$ \\
BEATs iter3 (Tmean, GPU)         & 6144 & $13.60 \pm 0.14$ & $17.91 \pm 0.14$  & $45.01/90.71$ \\
\bottomrule
\end{tabular}%
}
\end{table}

Table~\ref{tab:cost_frontends} reports per-clip inference cost from waveform to anomaly score. Handcrafted front ends were inexpensive on CPU: Log-Mel and MFCC required single-digit milliseconds end-to-end, and LPC-based features required tens of milliseconds. BEATs was about two orders of magnitude slower on CPU due to neural-encoder inference (45.01\,G MACs, 90.71\,M parameters) and reached comparable latency only on GPU. Together with the main benchmark results in Table~\ref{tab:main_without_label}, these numbers show that handcrafted representations provide a strong efficiency baseline for CPU-only training-free deployments. Complementary feature behavior under varying noise conditions is analyzed on MIMII in Appendix~Sec.~\ref{app:mimii_analysis}.


%% file: results/main_finetuning.tex
\begingroup
\renewcommand{\std}[1]{\unskip\hspace{0.08em}{\tiny$\pm$\,#1}}
\providecommand{\blockbest}[1]{\underline{#1}}
\providecommand{\best}[1]{\textbf{\underline{#1}}}
\begin{table}[!t]
\centering
\caption{Comparison of fine-tuned frontends using auxiliary labels with previously reported task-specific systems on DCASE2020T2, DCASE2023T2, and DCASE2024T2. Our scores are averaged over five independent trials. Underlining marks the best result within each comparison block, and bold underlining the best across all blocks.}
\label{tab:main_with_label}
\setlength{\tabcolsep}{2.5pt}
\resizebox{\columnwidth}{!}{%
\begin{tabular}{l@{\hskip 0.3cm}ll@{\hskip 0.5cm}ll@{\hskip 0.5cm}ll}
\toprule
& 
\multicolumn{2}{c}{DCASE2020T2} & 
\multicolumn{2}{c}{DCASE2023T2 (DG)} & 
\multicolumn{2}{c}{DCASE2024T2 (DG)} \\
\cmidrule(lr){2-3} \cmidrule(lr){4-5} \cmidrule(lr){6-7}
 Method & \shortstack{sec-dev\\{\scriptsize (Amean)}} & \shortstack{sec-eval\\{\scriptsize (Amean)}} & \shortstack{sec-dev\\{\scriptsize (Hmean)}} & \shortstack{sec-eval\\{\scriptsize (Hmean)}} & \shortstack{sec-dev\\{\scriptsize (Hmean)}} & \shortstack{sec-eval\\{\scriptsize (Hmean)}} \\
\midrule
\rowcolor{sectionrow}\multicolumn{7}{l}{Prior work with global matching design} \\
\addlinespace[2pt]
SSL4ASD w/ LDN \cite{local_den_norm_kevin, wilkinghoff2025density} &90.7 &90.2 &68.4 &68.0 &62.0 &54.7 \\
\hspace{0.5em}+ 10 ensembles &\blockbest{94.2} &93.3 &\best{71.3} &72.4 &\best{65.2} &56.5 \\
AnoPatch \cite{jiang_anopatch_2024}(BEATs Full FT) &90.9 &\blockbest{94.3} &64.2 &74.2 &62.5 \std{0.8} &65.6 \std{1.1} \\
\hspace{0.5em}+ pseudo labeling \cite{10889514} &- &- &- &- &64.1 \std{0.6} &\best{66.0} \std{0.7} \\
AnoPatch-LoRA \cite{anopatch_lora} &- &- &67.3 &\best{77.2} &- &- \\
UnlabeledASD \cite{10890020} &- &- &67.2 \std{0.7} &68.8 \std{0.9} &- &- \\
ASDKit \cite{fujimura2025asdkit} (BEATs-LoRA) &90.6 \std{0.2} &93.5 \std{0.3} &62.7 \std{0.7} &71.5 \std{0.5} &59.5 \std{0.9} &61.7 \std{0.7} \\
\midrule
\rowcolor{sectionrow}\multicolumn{7}{l}{This work (BEATs-iter3 Full FT)\textsuperscript{*}} \\
\addlinespace[2pt]
\multicolumn{7}{c}{\textbf{\textit{Normalization: LDN (default)}}} \\
Global matching baseline &91.6 \std{0.2} & 93.0 \std{0.6} & 59.2 \std{0.5} & 71.9 \std{0.2} & 62.0 \std{1.2} & 57.8 \std{0.5}\\
\hspace{0.5em}+ DMM & 94.6 \std{0.1} & 95.2 \std{0.3} & 61.1 \std{0.1} & 72.0 \std{0.5} & 62.4 \std{0.4} & 57.0 \std{0.5}\\
BEAM (proposed) &93.0 \std{0.3} &94.0 \std{0.6} & 61.7 \std{0.4} & \blockbest{75.3} \std{0.2} & \blockbest{63.5} \std{0.5} & 59.4 \std{0.6}\\
AdaBEAM (proposed) &\best{95.6} \std{0.1} & \best{96.2} \std{0.2} & \blockbest{63.1} \std{0.6} & 74.1 \std{0.6} & 62.9 \std{0.4} & \blockbest{59.9} \std{0.7}\\
\midrule
\multicolumn{7}{c}{\textbf{\textit{Normalization: Domain-wise standardized scoring}}} \\
Global matching baseline &92.0 \std{0.1} &\blockbest{93.9} \std{0.3} &62.1 \std{0.3} &75.5 \std{0.7} &61.0 \std{0.7} &63.1 \std{0.4} \\
\hspace{0.5em}+ DMM  &92.0 \std{0.1} &93.7 \std{0.4} &62.5 \std{0.3} &75.4 \std{1.3} &61.9 \std{0.5} &\blockbest{63.5} \std{0.3} \\
BEAM (proposed) &\blockbest{92.3} \std{0.1} &93.8 \std{0.3} &62.8 \std{0.1} &\blockbest{76.4} \std{0.5} &62.2 \std{0.4} &62.9 \std{0.2} \\
AdaBEAM (proposed) &92.2 \std{0.2} &93.4 \std{0.5} &\blockbest{63.4} \std{0.3} &75.7 \std{0.7} &\blockbest{62.9} \std{0.4} &63.4 \std{0.2} \\

\bottomrule

\addlinespace[2pt]
\multicolumn{7}{l}{\scriptsize\textsuperscript{*}Averaged over five independent trials.} \\

\end{tabular}
}
\end{table}
\endgroup

%% file: results/table_vi_scoring_ablation_official_dev_eval.tex
\begin{table}[!t]
\centering
\caption{Scoring ablation averaged across the ten sec-dev and sec-eval results of the five DCASE benchmarks ($n=10$). All sub-band methods use mean band aggregation unless specified in the final block; Mean, Max, and Min for DMM and AdaBEAM denote temporal-view fusion. The Average column excludes LPC. The symbol ``-'' denotes a configuration that is not applicable, and bold marks the best result in each column.}
\label{tab:scoring_ablation_avg}
\setlength{\tabcolsep}{4pt}
\resizebox{\columnwidth}{!}{%
\begin{tabular}{lccccccc}
\toprule
Method & Log-Mel & MFCC & LPC & EAT & EAT-L & BEATs & \shortstack{Average\\{\scriptsize w/o LPC}} \\
\midrule
\rowcolor{lightgray}\multicolumn{8}{l}{\textit{Global matching}} \\
Tmean    & 59.8 & 60.0 & 62.4 & 63.2 & 63.4 & 65.6 & 62.4 \\
Tmax     & 59.4 & 58.1 & - & 63.5 & 63.9 & 66.5 & 62.3 \\
DMM (Min)  & 60.0 & 59.7 & - & 64.3 & 64.4 & 66.5 & 63.0 \\
DMM (Max)  & 61.9 & 61.0 & - & 63.9 & 64.5 & 66.4 & 63.5 \\
DMM (Mean) & 61.7 & 60.9 & - & 63.7 & 64.1 & 66.1 & 63.3 \\
\midrule
\rowcolor{lightgray}\multicolumn{8}{l}{\textit{Sub-band matching (mean band aggregation)}} \\
BEAM-Tmean & 62.5 & 62.3 & \textbf{64.4} & 67.6 & 67.5 & 69.1 & 65.8 \\
BEAM-Tmax  & 59.0 & 56.4 & - & 67.3 & 67.9 & 69.4 & 64.0 \\
AdaBEAM (Min)  & 61.1 & 59.6 & - & \textbf{67.9} & \textbf{68.2} & 69.4 & 65.2 \\
AdaBEAM (Max)  & 63.5 & 62.4 & - & 67.7 & 68.0 & \textbf{69.7} & \textbf{66.3} \\
AdaBEAM (Mean) & \textbf{63.7} & \textbf{62.8} & - & 67.7 & 67.8 & 69.5 & \textbf{66.3} \\
\midrule
\rowcolor{lightgray}\multicolumn{8}{l}{\textit{Alternative band aggregation for BEAM-Tmean (Mean shown above)}} \\
Maximum & 61.0 & 61.0 & 62.8 & 65.9 & 66.9 & 67.5 & 64.5 \\
Minimum & 58.7 & 59.0 & 61.7 & 63.9 & 64.6 & 64.7 & 62.2 \\
\bottomrule
\end{tabular}
}
\end{table}

%% file: appendix13.tex
This appendix provides the proofs of Proposition~1 and Theorem~1, together with additional analyses of cross-benchmark band-width transfer, distance-metric sensitivity, and clip-level feature behavior on MIMII.

\subsection{Proof of Proposition~1}
\label{app:variance}

\begin{proof}
From the definition of $P$ in Eq.~\eqref{eq:P_main},
$S_{\mathrm{uni}}=S_{\mathrm{sub}}+P$. Therefore,
\begin{align}
\mathrm{Var}_{\mathcal N}(S_{\mathrm{sub}})
={}&
\mathrm{Var}_{\mathcal N}(S_{\mathrm{uni}})
-\mathrm{Var}_{\mathcal N}(P)
\nonumber\\
&-2\mathrm{Cov}_{\mathcal N}(S_{\mathrm{sub}},P).
\label{eq:app_var_identity}
\end{align}
By Eq.~\eqref{eq:lambda_main},
\[
-2\mathrm{Cov}_{\mathcal N}(S_{\mathrm{sub}},P)
\le
\lambda\,\mathrm{Var}_{\mathcal N}(S_{\mathrm{glob}}).
\]
The law of total variance gives
\begin{align}
\mathrm{Var}_{\mathcal N}(S_{\mathrm{glob}})
\ge{}&
\mathbb E_{\mathcal N}\!\left[
\mathrm{Var}_{\mathcal N}(S_{\mathrm{glob}}\mid G)
\right]
\nonumber\\
={}&
g_0\,
\mathbb E_{\mathcal N}\!\left[
\mathrm{Var}_{\mathcal N}(S_{\mathrm{uni}}\mid G)
\right],
\label{eq:tv_var_glob}
\end{align}
where the equality follows from Eq.~\eqref{eq:g0_main}. Similarly,
\begin{align}
\mathrm{Var}_{\mathcal N}(S_{\mathrm{uni}})
={}&
\mathbb E_{\mathcal N}\!\left[
\mathrm{Var}_{\mathcal N}(S_{\mathrm{uni}}\mid G)
\right]
\nonumber\\
&+
\mathrm{Var}_{\mathcal N}\!\left(
\mathbb E_{\mathcal N}[S_{\mathrm{uni}}\mid G]
\right)
\nonumber\\
={}&
(1+\gamma)
\mathbb E_{\mathcal N}\!\left[
\mathrm{Var}_{\mathcal N}(S_{\mathrm{uni}}\mid G)
\right],
\label{eq:tv_var_uni}
\end{align}
where the final equality follows from Eq.~\eqref{eq:gamma_main}.
Consequently,
\[
\mathrm{Var}_{\mathcal N}(S_{\mathrm{uni}})
\le
\frac{1+\gamma}{g_0}\,
\mathrm{Var}_{\mathcal N}(S_{\mathrm{glob}}).
\]
Substituting this bound into Eq.~\eqref{eq:app_var_identity}, using
Eq.~\eqref{eq:p0_main}, and applying the covariance bound above yields
\[
\mathrm{Var}_{\mathcal N}(S_{\mathrm{sub}})
\le
\left(
\frac{1+\gamma}{g_0}+\lambda-p_0
\right)
\mathrm{Var}_{\mathcal N}(S_{\mathrm{glob}}).
\]
This is precisely Eq.~\eqref{eq:prop1_main}.
Because the left-hand side is nonnegative and
$\mathrm{Var}_{\mathcal N}(S_{\mathrm{glob}})>0$, the same inequality
also implies $C_{\mathrm{var}}\geq0$; no separate nonnegativity
assumption on $C_{\mathrm{var}}$ is required.
\end{proof}

\paragraph{Interpretation of Proposition~1.}
The bound separates the two changes in the path $S_{\mathrm{glob}}\rightarrow S_{\mathrm{uni}}\rightarrow S_{\mathrm{sub}}$. The factor $(1+\gamma)/g_0$ describes the aggregation effect: $g_0>1$ indicates that uniform aggregation reduces within-regime normal-score variance, whereas larger $\gamma$ indicates stronger score drift between operating regimes. The term $\lambda-p_0$ describes the contribution of band-wise reference selection to the bound. Because $P=S_{\mathrm{uni}}-S_{\mathrm{sub}}\geq0$, $P$ is the penalty induced by forcing every band to use the globally selected reference; $p_0$ measures the normal-class variability of this penalty, while $\lambda$ conservatively accounts for its covariance with $S_{\mathrm{sub}}$. Thus, $C_{\mathrm{var}}$ is a data-dependent upper bound on the normal-score variance ratio. The proposition does not claim universal variance reduction; rather, $0<C_{\mathrm{var}}<1$ provides a sufficient condition under the observed score statistics.

\subsection{Proof of Theorem~1}
\label{app:dprime_proof}

\begin{proof}
Using Eq.~\eqref{eq:dprime_main},
\[
\frac{d'(S_{\mathrm{sub}})}{d'(S_{\mathrm{glob}})}
=
\frac{\Delta(S_{\mathrm{sub}})}
{\Delta(S_{\mathrm{glob}})}
\sqrt{
\frac{\mathrm{Var}_{\mathcal N}(S_{\mathrm{glob}})}
{\mathrm{Var}_{\mathcal N}(S_{\mathrm{sub}})}
}.
\]
Proposition~1 implies
\[
\sqrt{
\frac{\mathrm{Var}_{\mathcal N}(S_{\mathrm{glob}})}
{\mathrm{Var}_{\mathcal N}(S_{\mathrm{sub}})}
}
\ge
\frac{1}{\sqrt{C_{\mathrm{var}}}},
\]
where $C_{\mathrm{var}}>0$ follows because
$\mathrm{Var}_{\mathcal N}(S_{\mathrm{sub}})>0$,
$\mathrm{Var}_{\mathcal N}(S_{\mathrm{glob}})>0$, and
Eq.~\eqref{eq:prop1_main} would otherwise be impossible.
Moreover, because $\Delta(S_{\mathrm{glob}})>0$,
Eq.~\eqref{eq:theorem_condition_main} gives
\[
\frac{\Delta(S_{\mathrm{sub}})}
{\Delta(S_{\mathrm{glob}})}
\ge
\sqrt{C_{\mathrm{var}}}.
\]
Multiplying these inequalities proves
Eq.~\eqref{eq:theorem_result_main}.
\end{proof}

\paragraph{Interpretation of Theorem~1.}
Variance reduction alone does not guarantee better anomaly discrimination because sub-band matching may also reduce the anomalous--normal mean gap. Proposition~1 bounds the normal standard-deviation ratio by $\sqrt{C_{\mathrm{var}}}$; Theorem~1 therefore guarantees improved or equal $d'$ when the mean-gap ratio is at least this value. This sufficient condition conservatively ensures that any mean-gap loss does not offset the certified variance reduction. This explains why overly fine partitions can be harmful: they may strongly stabilize normal scores while discarding anomaly-relevant structure. Conversely, a moderate variance reduction can improve $d'$ when the mean separation is preserved or increased.

The results remain conditional and data-dependent. Proposition~1 certifies a strict reduction in normal-score variance when $0<C_{\mathrm{var}}<1$, while Theorem~1 guarantees improved or equal $d'$ under the stated mean-gap condition.

\subsection{Cross-Benchmark Band-Width Transfer}
\label{app:loyo_band_selection}

Table~\ref{tab:loyo_band_selection} evaluates whether the sub-band width $C$ transfers across benchmarks. Leave-One-Year-Out (LOYO) selection closely matched year-specific development tuning, with average sec-eval improvements of 3.1 and 2.9 percentage points, respectively. The selected widths were also stable across folds: MFCC used $C=20$, LPC spectrum used $C=3200$, and all deep features used one of eight frequency-patch groups in every fold. Log-Mel used $C=38$ in four folds and $C=25$ only when DCASE2023T2 was held out. The effect was feature-dependent: LOYO selection increased the mean gain most clearly for Log-Mel ($+2.9$ versus $+2.0$), was unchanged for EAT, and slightly reduced the gain for LPC spectrum ($+1.9$ versus $+2.1$). The confidence intervals overlap, so the small average advantage of LOYO ($+3.1$ versus $+2.9$) does not establish a universal improvement. Nevertheless, it suggests that selecting $C$ from several heterogeneous development benchmarks can reduce sensitivity to the specific characteristics of a single sec-dev set and thereby limit selection overfitting.

\begin{table}[!t]
\centering
\caption{Mean sec-eval benchmark-score improvement over global Tmean matching in percentage points, with 95\% paired bootstrap confidence intervals. Year-specific selection chooses $C$ on the corresponding sec-dev set. LOYO selection chooses $C$ using only the sec-dev sets of the other years and applies it to the held-out sec-eval set.}
\label{tab:loyo_band_selection}
\resizebox{0.89\columnwidth}{!}{%
\begin{tabular}{lcc}
\toprule
Feature & Year-specific & LOYO \\
\midrule
Log-Mel   & $+2.0\;[-0.4,+4.8]$ & $+2.9\;[+0.6,+5.4]$ \\
MFCC      & $+1.7\;[+0.9,+2.4]$ & $+1.8\;[+1.1,+2.4]$ \\
LPC spectrum & $+2.1\;[+0.7,+3.4]$ & $+1.9\;[+1.4,+2.8]$ \\
EAT       & $+4.7\;[+2.7,+6.8]$ & $+4.7\;[+2.7,+6.8]$ \\
EAT-large & $+4.2\;[+2.6,+6.0]$ & $+4.4\;[+2.5,+6.2]$ \\
BEATs iter3 & $+2.7\;[+1.9,+4.0]$ & $+2.9\;[+2.1,+4.0]$ \\
\midrule
\textbf{Average} & $\mathbf{+2.9\;[+1.7,+4.2]}$ & $\mathbf{+3.1\;[+2.0,+4.3]}$ \\
\bottomrule
\end{tabular}%
}
\end{table}

\subsection{Distance-Metric Sensitivity}
\label{app:distance_metric}

Table~\ref{tab:distance_metric_ablation} evaluates sensitivity to cosine and Euclidean distance while holding all other settings fixed. The raw global-matching averages were identical (60.5), and \textit{BEAM}+LDN improved over global matching with LDN for every frontend under both metrics. The overall \textit{BEAM}+LDN scores were also similar: 63.8 with Euclidean distance and 64.6 with cosine distance, compared with global averages of 61.0 and 61.4, respectively. Thus, the gains were not strongly dependent on the distance metric, although the preferred metric varied slightly by representation.

\begin{table}[!t]
\centering
\caption{Distance-metric sensitivity on sec-dev, averaged across the five DCASE benchmarks. Each cell reports Euclidean/Cosine, and Average is computed across the six frontends. Bold marks the best result in each row.}
\label{tab:distance_metric_ablation}
\setlength{\tabcolsep}{3pt}
\resizebox{0.76\columnwidth}{!}{%
\begin{tabular}{lcccc}
\toprule
Feature & \shortstack{Global\\w/o LDN} & \shortstack{\textit{BEAM}\\w/o LDN} & \shortstack{Global\\w/ LDN} & \shortstack{\textit{BEAM}\\w/ LDN} \\
\midrule
Log-Mel & 57.1/57.9 & 57.1/58.2 & 58.0/58.6 & 58.9/\textbf{61.0} \\
MFCC    & 57.6/58.4 & 57.6/57.9 & 58.3/58.7 & 59.2/\textbf{61.5} \\
LPC     & 62.8/60.9 & 62.8/62.0 & 62.5/61.8 & \textbf{64.0}/\textbf{64.0} \\
EAT     & 61.4/61.1 & 62.5/62.5 & 61.8/62.2 & \textbf{66.4}/66.2 \\
EAT-L   & 61.5/61.6 & 63.0/63.1 & 62.4/62.7 & 66.2/\textbf{66.3} \\
BEATs   & 62.8/63.0 & 64.4/63.8 & 63.2/64.2 & 68.1/\textbf{68.3} \\
\midrule
Average & 60.5/60.5 & 61.2/61.3 & 61.0/61.4 & 63.8/\textbf{64.6} \\
\bottomrule
\end{tabular}%
}
\end{table}

\subsection{Clip-Level Feature Analysis on MIMII}
\label{app:mimii_analysis}

MIMII contains normal and anomalous recordings of Fan, Pump, Slider, and Valve machines in real factory environments at SNRs of $-6$, $0$, and $6$~dB~\cite{purohit_mimii_2019}. Fan and Pump are primarily stationary, whereas Slider and Valve are more non-stationary. We use the train/test split of~\cite{regional_local} and report the arithmetic mean of AUC and pAUC after averaging each metric across machine sections.

Figure~\ref{fig:mimii_compare_feats} compares clip-level front ends using tied-reference global matching (standard $k$-NN without local density normalization~\cite{10887974}). Among the evaluated handcrafted features, the LPC spectrum performed particularly well and, in this setting, even surpassed deep BEATs features. Temporal max pooling also consistently outperformed temporal mean pooling, highlighting the importance of retaining transient information for non-stationary machines.

\begin{figure}[!t]
    \centering
    \includegraphics[width=0.97\columnwidth]{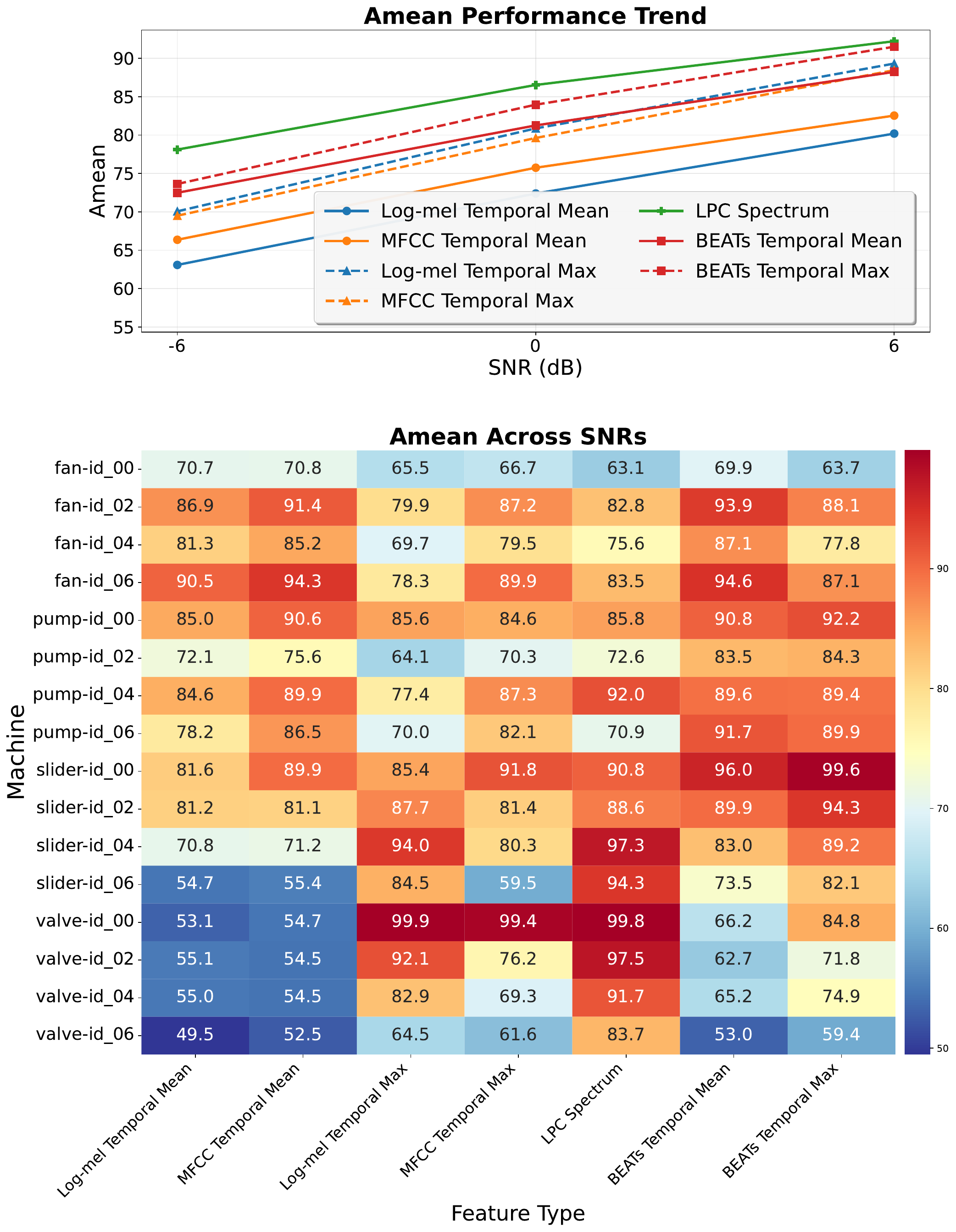}
    \caption{Comparison of clip-level features on MIMII using tied-reference global $k$-NN without LDN. Top: arithmetic mean of AUC and pAUC (Amean) at each SNR, averaged across machine sections. Bottom: Amean for each machine section, averaged across the three SNR conditions.}
    \label{fig:mimii_compare_feats}
\end{figure}

These trends can be explained by how LPC relates to conventional spectral features. Like max-pooled Log-Mel or MFCCs, LPC emphasized prominent spectral peaks while modeling a smooth spectral envelope that captured longer-term structure in stationary sounds. At the same time, because LPC estimated a global envelope, it could be more sensitive to background noise in low-SNR stationary conditions, where band-wise energy integration and smoothing with Log-Mel or MFCCs could be more robust. Overall, LPC-based representations offered a complementary trade-off: they could outperform pooled Log-Mel and MFCC baselines while remaining lightweight, non-parametric, and interpretable by revealing frequency regions where anomalies occurred.